\documentclass[12pt,a4paper,oneside]{article}

\usepackage{fancybox,prodint}
\usepackage[authoryear,round]{natbib}
\usepackage{amssymb}
\usepackage{dsfont}
\usepackage{leftindex}
\usepackage{enumitem}

\newtheorem{lemma}{Lemma}
\newtheorem{corollary}{Corollary}
\newtheorem{proof}{Proof}

\newcounter{AnnahmenCounter}

\newcommand{\qed}{\hfill $\Box$\\}

\newcounter{AnnahmenCounterLocal}
\newcounter{AnnahmenCounterGlobal}
\newcounter{AnnahmenCounterA}
\newcounter{AnnahmenCounterB}

\newcounter{AnnahmenCounterP}
\newcounter{AnnahmenCounterZ}
\newcounter{AnnahmenCounterM}
\newcounter{AnnahmenCounterK}
\newcounter{AnnahmenCounterKa}

\newcommand{\AnnA}{A}

\newcommand{\AnnANumm}{(\AnnA\arabic{*})}

\newtheorem{theorem}{Theorem}
\begin{document}

\title{Left-truncated discrete lifespans: The AFiD enterprise panel}

\author{Eric Scholz \& Rafael Wei{\ss}bach \\[2mm] \textit{\footnotesize{Chair of Statistics and Econometrics,}}   \\[-2mm]
	\textit{\footnotesize{  Faculty for Economic and Social Sciences,}} \\[-2mm]
	\textit{\footnotesize{University of Rostock}} \\
}
\date{ }
\maketitle

\renewcommand{\baselinestretch}{1.3}\normalsize

\begin{abstract}
Our model for the lifespan of an enterprise is the geometric distribution. We do not formulate a model for enterprise foundation, but assume that foundations and lifespans are independent. We aim to fit the model to information about foundation and closure of German enterprises in the AFiD panel. The lifespan for an enterprise that has been founded before the first wave of the panel is either left truncated, when the enterprise is contained in the panel, or missing, when it already closed down before the first wave. Marginalizing the likelihood to that part of the enterprise history after the first wave contributes to the aim of a closed-form estimate and standard error. Invariance under the foundation distribution is achived by conditioning on observability of the enterprises. The conditional marginal likelihood can be written as a function of a martingale. The later arises when calculating the compensator, with respect some filtration, of a process that counts the closures. The estimator itself can then also be written as a martingale transform and consistency as well as asymptotic normality are easily proven. The life expectancy of German enterprises, estimated from the demographic information about 1.4 million enterprises for the years 2018 and 2019, are ten years. The width of the confidence interval are two months. Closure after the last wave is taken into account as right censored.
 \\[2mm]
\end{abstract}

\section{Introduction} \label{introduc}

Our aim is to estimate the probability for an enterprise to close down - by insolvency or any other reason - over the course of the next year. For each of several enterprises we will possess information over several successive years. The natural approach is to count the experiments and relate them to the count of negative events, namely of closures. We will study the numerator and the denominator of this ratio extensively. We will assume that the closure probability neither changes over time nor with age. There are two questions that arise in this approach. First, the survival experiments for an enterprise will typically not be observed from its foundation onward, but are counted from some specific age onward. What is the impact of ignoring the survival before that age, i.e. of forgetting it? Second, the number of survival experiments will not be the same for all enterprises, because some will be founded earlier and some later, and our panel starts observation for all at the same time. Assuming the independence among enterprises, how can the asymptotic standard error and asymptotic normality still be derived with a limit theorem that obviously must obey a Lindeberg-L\'{e}vy condition due to the different variances?     

 Our lifetime model measures enterprise survival as rounds of a game, the `years' (called `epochs' in \cite{Fel1}) rather than considering time as real-valued. The reason is that the official statistics records survival annually. `Age' is the counted years from foundation.

\subsection{Data: The AFiD Panel} \label{sec11}
The German Statistical Office (DESTATIS) stores ``Amtliche Firmendaten f\"ur Deutschland'' (translates as ``official enterprise data for Germany'') annually. With effect of 2018, as reporting year, DESTATIS has started their reports on business activity with statistical unit in the enterprise definition by the European Commission. DESTATIS has disclosed enterprise closures in the years 2018 and 2019, as well as information on enterprises that survived 2019 (Table \ref{Data}). The data is a panel, by record linkage information on equal enterprises for different years allows to the identify individual path for each enterprise. 
The sampling design is retrospective because foundation years before 2018 are attributed. For the ease of methodological display we limit interest to population of enterprises characterized by foundation in the five ($G$) years 2013--2017.
\begin{table}[b]	
	\caption{Counts of German enterprises closed down 2018 or later. Recorded enterprises have been founded in  the years 2013--2017.} \label{Data} 
	\begin{center}
		\begin{tabular}{cccc}
			\hline \hline 
		Total no. of & 	\multicolumn{3}{c}{Thereof with year of closure ...}    \\
		observations ($m$) &	2018 ($d^{obs}_j=1$) &      2019 ($d^{obs}_j=2$) &  2020 or later ($m_{cens}$) \\ \hline \hline
	1,447,814	& 	168,112 & 107,050 & 1,172,652  \\ \hline \hline
		\end{tabular} \\
		(Explanation of numbers and symbols is distributed over larger parts of text.)
	\end{center}
\end{table}
The data in Table \ref{Data} contains left-truncated lifespans because lifetimes before 2018 are not observed. Enterprises that closed down before 2018 are unrecorded. (Enterprises closed after 2019 are right-censored.) \cite{rinkseif} and references therein contain a detailed description of the data. We use the attributes urs$\_$we$\_$beginn$\_$datum, as foundation year, and urs$\_$we$\_$ende$\_$datum, as year of closure (with label `3000' for right-censoring). Note that if the table did represent, prospectively, for enterprises founded in 2013--2017, numbers of closures in 2018 and later, the analysis would be easier.  

Because our model disregards the possibility of the closure probability to change with age, the knowledge of the foundation year for each enterprise will be redundant. With possession of the foundation year, more complex models can be fitted to the data than ours, by using the same method that we will describe. 

\subsection{Literature review}

Cross-sectional statistics for business survival mostly concentrates on new firms \cite[see][for China, Germany, Japan, Portugal and the USA]{cochran1981,Mata1994,Audretsch1995,Honjo2000,rinkseif,kato2023,pittiglio2023}. These studies do not include retrospective conclusions and an early conclusive attempt is \cite{bruederl1992} who oversample  their retrospectives interviews with founder of closed enterprises, as compared interviews with active (and partly right-censored) enterprises. Life expectancy as measurement in business demographic has been studied in \cite{Reis2015} for Portuguese enterprises. Recently results on truncation (however for double-truncation) with application to business survival contain \cite{jacobo2024} for Spain and \cite{topaweis2024} for Germany. The study at hand is a time-discrete version of the influential study `Mortality of Diabetics in the County of Fyn' in \cite{And}. Left-truncation is known to be a typical design defect in multi-state models also by \cite{putter2006} and \cite{michemura24}. Ignoring truncation has been found to introduce a `immortal time bias' to the analysis \cite[see e.g.][]{HERNAN201670,yadav2021}.

	\section{Longitudinal dependence and memory}	\label{ltrcfixed} 
	
We will assume that in a set of enterprises, these develop independent or each other. With interest in the probability that an enterprise survives the next year, $1-\theta$, we could  consider only the first year after foundation of each enterprise. A simple random sample of Bernoulli experiments would be the respective data design. Integrating information about subsequent years for the enterprises, we cannot assume all survival experiments to form a simple random sample, because experiments for the same enterprise are for sure longitudinally dependent. Additionally, the number of experiments for each enterprise, i.e. the lifespan, is random and random sample sizes are not conventional for elementary inference.  	

The lifespan of an enterprise, $X$, i.e. the duration from its foundation to its closure, see e.g. the middle path in Figure \ref{model}, will become central. For simplicity of the display we do not assume the annual probability to close business to vary with the age of an enterprise, i.e. we assume a geometric distribution for the lifespan. 

\begin{figure}[htb!] \centering
	\setlength{\unitlength}{1.2cm}
	\begin{picture}(11.5,5.5)
		\linethickness{0.3mm}
		\put(0.1,1.65){\vector(1,0){9.5}} \put(9.0,1.2){time}
		\put(0.5,1.6){\line(0,1){0.1}} 
		\put(1.75,1.6){\line(0,1){0.1}}
		\put(3.1,1.6){\line(0,1){0.1}}
		\put(4.25,1.6){\line(0,1){0.1}}
		\put(5.5,1.6){\line(0,1){0.1}}
		\put(6.75,1.6){\line(0,1){0.1}}
		\put(8,1.6){\line(0,1){0.1}}
		
		\put(6.75,1.6){\line(0,1){3.7}} 
		\put(8,1.6){\line(0,1){3.7}}

		\put(0,1.15){2013}
		\put(2.5,1.15){2015}
		\put(5,1.15){2017}
		\put(6.25,1.15){2018}
		\put(7.5,1.15){2019}
		\put(8.5,1.15)
		
		\linethickness{0.15mm}

		\put(0.45,5.02){$\diamond$} \put(1.75,5.1){\circle{0.15}} \put(3,5.1){\circle*{0.15}}
		\multiput(0.6,5.1)(0.4,0){3}{\line(1,0){0.2}}
		\multiput(1.85,5.1)(0.4,0){3}{\line(1,0){0.2}}

		\put(1.7,4.44){$\diamond$} \put(3,4.5){\circle{0.15}} \put(4.25,4.5){\circle{0.15}} \put(5.5,4.5){\circle{0.15}} \put(6.75,4.5){\circle*{0.15}}
		\multiput(1.85,4.5)(0.4,0){3}{\line(1,0){0.2}}
		\multiput(3.1,4.5)(0.4,0){3}{\line(1,0){0.2}}
		\multiput(4.35,4.5)(0.4,0){3}{\line(1,0){0.2}}
		\multiput(5.6,4.5)(0.4,0){3}{\line(1,0){0.2}}

		\put(6.8,2.65){\line(1,0){1.15}}

		\put(1.7,2.61){$\diamond$} \put(3,2.65){\circle{0.15}} \put(4.25,2.65){\circle{0.15}} \put(5.5,2.65){\circle{0.15}} \put(6.75,2.65){\circle{0.15}} \put(8,2.65){\circle{0.15}}
		\multiput(1.85,2.65)(0.4,0){3}{\line(1,0){0.2}}
		\multiput(3.1,2.65)(0.4,0){3}{\line(1,0){0.2}}
		\multiput(4.35,2.65)(0.4,0){3}{\line(1,0){0.2}}
		\multiput(5.6,2.65)(0.4,0){3}{\line(1,0){0.2}}
		\multiput(6.85,2.65)(0.4,0){3}{\line(1,0){0.2}}
		\multiput(8.1,2.65)(0.4,0){4}{\line(1,0){0.2}}
		
		\put(1.05,4.4){$\underbrace{\vphantom{ } \hspace{4.5cm}}_{Age-one-year-before-observation-begin \,  T \, (t)}$}
		\put(1.75,3.6){$\underbrace{\vphantom{ } \hspace{6cm}}_{Lifespan \,  X}$}
		\put(1.75,2.55){$\underbrace{\vphantom{ } \hspace{7.5cm}}_{Age-at-observation-end \,  T \; (t)+2}$}
		\put(0.5,1.05){$\underbrace{\vphantom{ } \hspace{6 cm}}_{foundation \,  period}$}
		\put(6.4,1.05){$\underbrace{\vphantom{ } \hspace{1.5cm}}_{observation \,  period}$}
		
	\end{picture}
	\caption{Three cases of foundation  (white diamond), year of closure (black circle), observed duration (dashed line with white bullets) and unobserved duration (dashed line).
		(Top: left-truncated) Path for enterprises founded in 2013 and closed before 2018, namely in 2015, i.e. with (middle: observable) Path for enterprises founded in 2014 and closed in the observation period, namely in 2018 (bottom: right-censored) path for enterprise founded in 2014 with closure after 2019} \label{model} (Explanation of panels and symbols is distributed over larger parts of text.)
\end{figure}
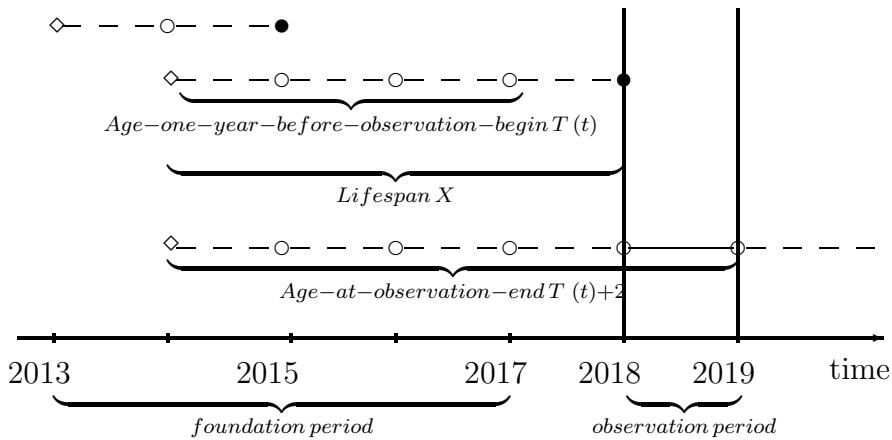

If lives of independent enterprises are followed from foundation until closure, be it that they are founded at the same time  or that foundations are scattered over time, a simple random sample assumption for the recorded $X$'s is appropriate. The maximum likelihood estimator for $\theta$ is then a ratio, where each enterprise contributes a numerical `1' to the numerator and the observed time at risk, $X$, to the denominator. Instead, our situation of Section \ref{sec11} is a panel and we derive now likelihood-based inference. In this Section, the study design is simplified to a fixed foundation year, i.e. formally to stratified sampling. This section sketches the main arguments and none of its results will be referenced later on. 
The next Section \ref{mortalran} then considers the realistic design with the foundation year as random. All arguments will there be given in mathematical rigour. 

	 Let $X$ be define on $(\Omega_X, \mathcal{F}, P_\theta)$ with probability density function (pdf) \citep[][Formula 10.71]{Jo0} 
	\begin{equation} \label{Grundmodell0}
	f_G(\mathbf{x})=\theta (1-\theta)^{\mathbf{x}-1}, \quad \mathbf{x} \in \mathbb{N}.
	\end{equation}

For simplicity of display we avoid that foundation takes place in the observation period (2018 and 2019) so that our population are enterprises founded 2017 or earlier (see Table \ref{Data} and Figure \ref{model}). That foundation before 2013 will later be ruled-out for practical reasons is theoretically without consequences. Instead of coding the foundation in calendar year, it will be useful to count years backwards from 2018 and such count $t$ is equal to the numbers of years elapsed from foundation until one year before observation starts. Figure \ref{model} illustrates $t$ with five foundation cohorts as population (middle path).
Also in the textbook situation of the linear regression one typically start with a pre-determined, fixed, covariate. Stratified sampling is rare for data in the social and economic sciences and the assumption is especially unrealistic for our application of the AFiD panel. We will continue in Section \ref{mortalran} with a random $T$ and, as in linear regression, with conditioning.  However, some arguments are easier to explain for deterministic $t$ (and are still true for random $T$). As in linear regression the distribution of the covariate $T$ will not require a specification so that the analysis is semi-parametric.   

The lifespan of an enterprise, $X$, is observed when closure has not yet occurred in 2017 (see Figure \ref{model}), more formally, 
\begin{equation}\label{preass3}
X \; \text{is observed if} \; X\ge t+1.
\end{equation}
Interestingly, a similar retrospective data collection lead \cite{bruederl1992} to `oversample' enterprises closed before the study start when sending questionnaires to their founders. We collect no information at all about those enterprises. The probability measure for the observed $X$ (or conditional on observation) is hence $P_{\theta}(\mathbb{G})/P_{\theta}(X\ge t+1)$ for $\mathbb{G} \in \sigma\{\Omega \setminus \{X\ge t+1\}\}$. 
Note that (i) the assumption of a deterministic $t$ implies its knowledge also for any unobserved enterprise, symbolized in the top path of Figure \ref{model}. Note (ii), that in this situation, the  size $n$ of the (latent) sample which is drawn from  the population of all enterprises. Note (iii) that the geometric distribution has no memory so that the conditional distribution is also a geometric distribution (with same parameter). We will refrain from using this property in order to enable easy generalization of the results.

For an enterprise that closed down before 2018 (top path in Figure \ref{model}), the unobserved survival experiments before 2018 are disregarded for a statistical analysis. Therefore, for an enterprise surviving 2018, its unobserved survival experiments (before 2018) should also be disregarded. The enterprise should contribute a `1' to the numerator and the time observably at risk of closure, $X-t$, to the denominator. For instance, for the middle path in Figure \ref{model} this results in $\mathbf{x}-t = 4 - 3 = 1$ year. The harder arguments for this intuitive bargain now follow.  (Note that, seeing the enterprise in 2018 at an age of $t+1$ implies $t$ positive survival experiments. A full maximum likelihood as in \cite{weiswied2021} would take account for all experiments.) 

We may express the $\mathbb{N}$-valued lifespan $X$ equivalently with the infinity vector `counting process' $N(x):=\mathds{1}_{\{X \le x\}}$, or its increments:
\begin{equation} \label{obsasproc}
	N:= \{N(x), x \in \mathbb{N}_0\} \; \text{or equivalently} \; \Delta N := \{\Delta N(x),  x \in \mathbb{N}\} 
	\end{equation} 
(Throughout the text, for whatever set $\mathbb{G} \subset \Omega$, $\mathds{1}_{\mathbb{G}}$ becomes one when $\omega \in \mathbb{G}$, e.g. if event has occurred, and is zero else. Also note that for whatever process $Z$, it is $\Delta Z(x):= Z(x)- Z(x-1)$ and here $N(0):=0$.)
We will see that we can confine the random vectors to a finite dimension. The probability measure $P_{\theta}$ for $X$ translates in obvious manner to one for $N$, and we use - a bit ambiguously - again $P_{\theta}$, with pdf $f^N$, alike for the $\sigma$-field $\mathcal{F}$. For inference, we could think of marginalizing the distribution of the vectors $N$ or $\Delta N$ to experiments starting with $t$ \cite[see][Definition 7.2(i)]{gourieroux1995}. However, with the temporal interpretation in mind, and as $t$ will become a random $T$ in the next Section \ref{mortalran}, we use the concept of forgetting. It can more easily be generalized, because $T$ will become a stopping time. 

\subsection{Marginalization and filtration}

We  describe a probability model for the vector $N$ in order to forget the survival experiments before the fixed $t$. It requires a concept of the `past', as opposed to  the presence and the future.
 
To be able to answer at an age $x$ every question `$X \le k$?'  for whatever $k \le x$ - i.e. to know at the age $x$ whether the enterprise is closed yet, and if, when, - defines the filtration consisting of $\mathcal{F}_x:=\sigma\{\mathds{1}_{\{X \le k\}}, 1 \le k  \le x\}$ ($\mathcal{F}_0:= \sigma \{\emptyset\}$). Note that $\mathcal{F}_x=\sigma\{\mathds{1}_{\{X = k\}}, 1 \le k  \le x\}$, however the inequality relation will later allow to the replace the indicators by counting processes. (Formally it is said, that for an enterprise randomly selected from the population, $X$ is defined on the stochastic basis $(\Omega_X, \mathcal{F}, \{\mathcal{F}_x, x \geq 0\}, P_\theta)$.) Of course $E_{\theta}(\Delta N(x)\vert \mathcal{F}_{x-1})=\theta \mathds{1}_{\{X \ge x\}}=: \Delta A(x,\theta)$ ($A(0,\theta):=0$) being a conventional concept of `location' for a stochastic process.

In the context of stochastic processes, the marginalization is reducing attention and corresponds to a coarser filtration. (Increasing the filtration will be necessary when the random $T$ introduces more information.) Initiating attention from some age $t$ onward (see also Figure \ref{model}), requires to exclude earlier outcomes from the probability model, formally represented by the set $\{ \emptyset, 1 \leq k \leq t\}$. To know at the age of $x \ge t+1$, whether the enterprise is closed yet ($\mathds{1}_{\{X \ge t+1\}}$), and if, when, - is now
\begin{equation*}
	\leftindex_t{\mathcal{G}}_x  :=  \sigma\{ \mathds{1}_{\{t +1 \le X \le k\}}, \mathds{1}_{\{X \ge k+1\}}, t+1 \le k \le x\}.
\end{equation*}	
Note that the processes $\mathds{1}_{\{X \ge k+1\}}$ can replace the information $\mathds{1}_{\{X \ge t+1\}}$ because $\mathds{1}_{\{X \ge k+1\}}$ for $k \ge t+1$ are in combination with $\mathds{1}_{\{t +1 \le X \le k\}}$ redundant. In view of Figure \ref{model} we can call this an {\em observed} filtration.

If $X \le t$ as well as $N(x)=1$ for $x \ge t$ (top path in Figure \ref{model}), no observable development will occur after $t$. We `gauge' it to zero by $\leftindex_t{N}(x):=N(x) - N(x \wedge t)=\mathds{1}_{\{t+1 \le X\}} \mathds{1}_{\{X \le x\}}$ (with $a \wedge b := min (a,b)$ for $a,b \in \mathbb{N}$), i.e. $\leftindex_t{N}$ is $N$ if $t+1 \leq X$ and constantly zero otherwise.  Note that this does not change the generated filtration $\sigma\{\leftindex_t{N}(k), \mathds{1}_{\{X \ge k+1\}}, t+1 \le k \le x\}=\leftindex_t{\mathcal{G}}_x$. (Note that we do not distinguish symbolically between a filtration $\{\mathcal{H}_x\}$ and its element, the $\sigma$-field $\mathcal{H}_x$, as long as the text supplies the distinction.)

Remember that a likelihood is a density -  similarly, the marginal likelihood is the marginal density, starting from $t$ - with respect to some dominating measure -  evaluated at the observed data. Now the $N$ in \eqref{obsasproc}, but also $\leftindex_t{N}$, can be regarded as vector of Bernoulli random variables with pdf $f^{\leftindex_t{N}}$. (The `marginalisation' by deciding to start observation of enterprise-specific experiments from $t$ - and later $T$ - onward, will not be mentioned from now on, most of the time. We will act as if given as data $\leftindex_t{N}$.) Consider now the longitudinal dependence in the vector. Without interference with the assumption of a geometric distribution for $X$, we can stop consideration about lifespans in general at an age of $\chi$, an unimportant large number. (That a discussion about $\chi$ is largely not necessary, will later become clear in Section \ref{sec2_2}.) The next Lemma will give two representations of the marginal pdf of $\leftindex_t{N}$. The first will allow an intuitive interpretation and the second will prepare for the estimator of $\theta$ to be analysed as martingale transform, especially it will simplify to derive the standard error.

 With respect to the second aim, note first that (for $x \ge t+1$)
  \begin{equation*}
  	\mathds{1}_{\{x \le X\}} \theta
  	\stackrel{(I)}{=}  Pr\{\Delta \leftindex_t{N}(x)=1 \vert \leftindex_t{\mathcal{G}}_{x-1}\} 
  	\stackrel{(II)}{=}  \mathbb{E}(\Delta \leftindex_t{N}(x) \vert \leftindex_t{\mathcal{G}}_{x-1}),
  \end{equation*}
  where (I) is because of the Markovian property of $\leftindex_t{N}$ and (II) is because $\Delta \leftindex_t{N}(x)$ is a Bernoulli random variable. (Note that by using $Pr$ and $\mathbb{E}$ we signal that we will need to discuss the exact dominating measure, but postpone the discussion to Section \ref{mortalran}.) By this (and called Doob decomposition),
  \begin{equation} \label{defmart}
  	\leftindex_t{N}(x) - \underbrace{\theta \sum_{k=t}^x \mathds{1}_{\{k \le X\}}}_{=:\leftindex_t{A}(x, \theta)}
  \end{equation}
  \citet[][Condition (35.1) of Section 35]{Billing} to be a $\leftindex_t{\mathcal{G}}_x$-martingale  is fulfilled (and $\leftindex_t{A}$ the `compensator' of $\leftindex_t{N}$ and hence $\leftindex_t{\mathcal{G}}_{x-1}$-measurable, i.e. deterministic at age $x-1$). (All other conditions of a martingale \citep[see][Section 35]{Billing} are also fulfilled. For the sake of brevity, their proofs will only be given for the random design.)

\begin{lemma} \label{lemmaeinfach}
Under Assumption \ref{Grundmodell0} and definitions \eqref{obsasproc}, the pdf of the enterprise life until $\chi$, conditional on the observation status $\{t+1 \le X\}$, and marginal from $t$, (with $0^0:=1$) is
\begin{multline} \label{likecontr2}
f^{\leftindex_t{N}\vert \mathds{1}_{\{t+1 \le X\}}}(\leftindex_t{n}\vert \mathds{1}_{\{t+1 \le \mathbf{x}\}}) \\ =  \mathds{1}_{\{t+1 \le \mathbf{x}\}} (1- \theta)^{\mathbf{x}-(t+1)} \theta + \mathds{1}_{\{t \ge \mathbf{x}\}} \\
=  \prod_{x=t+1}^{\chi}  (1 - \Delta \leftindex_t{a}(x, \theta))^{(1- \Delta \leftindex_t{n}(x))} (\Delta \leftindex_t{a}(x, \theta))^{\Delta \leftindex_t{n}(x)}, 
\end{multline}
where $\leftindex_t{n}$, $\leftindex_t{a}$ and $\mathbf{x}$ are the outcomes (and later realizations) of $\leftindex_t{N}$, $\leftindex_t{A}$ and $X$.
\end{lemma}

(Throughout the text and in extension to $\mathds{1}_{\mathbb{G}}$, $\mathds{1}_{condition}$ becomes one when the condition is true and is zero else.)
	
\begin{proof}
	Note that
	\begin{multline} \label{likecontr}
		f^{\leftindex_t{N}\vert \mathds{1}_{\{t+1 \le X\}}}(\leftindex_t{n}\vert \mathds{1}_{\{t+1 \le \mathbf{x}\}}) \\  \stackrel{(i)}{=}  \prod_{x=t+1}^{\chi} f^{\leftindex_t{N}(x) \vert \leftindex_t{N}(x-1),\mathds{1}_{\{t+1 \le X\}}}(\leftindex_t{n}(x) \vert \leftindex_t{n}({x-1}),\mathds{1}_{\{t+1 \le \mathbf{x}\}})  \\
		 \stackrel{(ii)}{=}   \prod_{x=t+1}^{\chi} (1 - \mathds{1}_{\{x \le \mathbf{x}\}} \theta)^{1- \Delta \leftindex_t{n}(x)} (\mathds{1}_{\{x \le \mathbf{x}\}} \theta)^{\Delta \leftindex_t{n}(x)}  
	\end{multline}
	
For (i). Both $\leftindex_t{N}$ and ${N}$ are obviously Markovian, e.g. $P_{\theta}(N(3)=b\vert N(2)=a_2, N(1)=a_1)=P_{\theta}(N(3)=b\vert N(2)=a_2)$. 

For (ii): It is for $x \in \mathbb{N}$: 
\begin{equation} \label{conprobs}
	P_{\theta}(N(x)=b \vert N(x-1) =a)	=\begin{cases} 1- \theta  & \text{for} \; a=0,b=0 \; (\Delta N(x)=0) \\
		\theta	&  \text{for} \; a=0,b=1 \; (\Delta N(x)=1)\\
		1	&  \text{for} \; a=1,b=1 \; (\Delta N(x)=0) \\
	\end{cases}
\end{equation}

For an observable lifespan (i.e. $\mathbf{x} \ge t+1$) and $x \ge t+1$ the conditional pdf $P_{\theta}(\leftindex_t{N}(x)=b \vert \leftindex_t{N}(x-1) =a, X \ge t+1)$ is as well \eqref{conprobs}. 
For an unobserved lifespan (i.e. $\mathbf{x} \le t$) it is: 
	\begin{equation*}
	P_{\theta}(\leftindex_t{N}(x)=b \vert \leftindex_t{N}(x-1) =a, X \le t)= 
	\begin{cases}
		1 & \text{for} \; (a,b)=(0,0) \;  (\Delta \leftindex_t{N}(x)=0)\\
		0 & \text{for} \; (a,b)=(0,1)  \;  (\Delta \leftindex_t{N}(x)=0)\\
		0 & \text{for} \; (a,b)=(1,1)  \;  (\Delta \leftindex_t{N}(x)=0)	
	\end{cases}	
\end{equation*}
The $\mathds{1}_{\{x \le \mathbf{x}\}}$ in the likelihood contribution prevents that further $(1- \theta)$'s are multiplied after closure has been reached. Note the longitudinal dependence, i.e. that $N(x)$ and $N(x-1)$ are not stochastically independent, is given by the absence of $(a,b)=(1,0)$ in  \eqref{conprobs} and will later prohibit a central limit theorem and require a martingale limit theorem.

Representation \eqref{likecontr} allows now the first line of \eqref{likecontr2} directly (with $0^0:=1$) and the second line (with \eqref{defmart}). \qed
\end{proof}

A proof which avoids the Markov argument and is hence easier generalized to multi-state models, uses the filtration $\leftindex_t{\mathcal{G}}$ explicitly and the equivalence \eqref{obsasproc}. The pdf $f^{\Delta \leftindex_t{N}, \leftindex_t{Y}}$ is then  (evaluated at $\Delta \leftindex_t{N}, \leftindex_t{Y}$) a product of factors such as  
\begin{eqnarray*}
P_{\theta} (\Delta \leftindex_t{N}(x)=1 \vert \leftindex_t{\mathcal{G}}_{x-1})^{\Delta \leftindex_t{N}(x)} & = & E_{\theta} (\Delta \leftindex_t{N}(x)\vert \leftindex_t{\mathcal{G}}_{x-1})^{\Delta \leftindex_t{N}(x)} \\
& = & (\Delta \leftindex_t{A}(x))^{\Delta \leftindex_t{N}(x)}.
\end{eqnarray*}

 We can now give the answer to the first question in Section \ref{introduc} as more rigorous reason for the estimation procedure described at the beginning of this Section \ref{ltrcfixed}.
 Note that the (log) conditional likelihood sum runs over all outcomes of the condition, here the observability status $\mathds{1}_{\{t+1 \le X\}}$ namely $\mathds{1}_{\{t+1 \le \mathbf{x}\}} =0$ and $\mathds{1}_{\{t+1 \le \mathbf{x}\}}=1$ \citep[see][Definition 7.2(ii)]{gourieroux1995}. The contribution $f^{\leftindex_t{N}\vert \mathds{1}_{\{t+1 \le X\}}}$ of a unobserved lifetime, $\mathds{1}_{\{t+1 \le X\}}=0$ (Figure \ref{model}, top) has a compensator of value zero, thereby ensuring that the likelihood contribution \eqref{likecontr2} is one. 
  Now for an observation, i.e. $\mathbf{x} \ge t+1$, it is by the first represenation of \eqref{likecontr2} 
 \begin{eqnarray*}
 	\frac{d}{d \theta}  \log f^{\leftindex_t{N}\vert \mathds{1}_{\{t+1 \le X\}}}(\leftindex_t{n}\vert \mathds{1}_{\{t+1 \le \mathbf{x}\}})  =  0 
 	\Leftrightarrow \theta = (\mathbf{x}-t)^{-1} \mathds{1}_{\{t+1 \le \mathbf{x}\}}
 \end{eqnarray*}
Omitting the enterprises that closed down before 2018 (in the estimator's numerator) in order to balance the reduction (in the denominator) to the `time at observable risk' of $\mathbf{x}-t$ years for those enterprises that survived 2017 maximizes the conditional likelihood.
 
Consistency of the approach will now be shown in the next main section. A first result will be that for random $T$ an enterprise with $X \le T$ will not be observed at all (not only its lifespan is not observed). Still, the its likelihood contribution will be one.  We will also see that in the likelihood contribution \eqref{likecontr2} (first line), we can simply replace $t$ by $T$. Before that, the next subsection ins concerned with those enterprises that survive 2019, the right boundary of observation period.

\subsection{Adjustment for right-censoring} \label{sec2_2}

Observation of any enterprise ends after the $s=2$ two years 2018 and 2019. Right-censored is a lifespan if the enterprise closures down after 2020 or later (see Figure \ref{model} (bottom path) and Table \ref{Data} (last column)). By superimposing right-censoring on left-truncation, the resulting process is the observed left-truncated and right-censored counting process and can be expressed by $\leftindex_t{N}^c(x) := \sum_{k=1}^x C(k) \Delta  \leftindex_t{N}(k)= \mathds{1}_{\{t+1 \le X \le x \wedge (t+2)\}}$ with $C(x):= \mathds{1}_{\{x \le t+2\}}$  \cite[compare][Example 1.4.2]{flem1991}.

The corresponding compensator is 
$\leftindex_t{A}^c(x, \theta) := \theta \sum_{k=1}^x \mathds{1}_{\{t \le k \le X \wedge (t+2) \}}$	
with to an observed filtration 
\[
\leftindex_t{\mathcal{F}}^c_x := \sigma\{\leftindex_t{N}^c(k), t+1 \le k \le x\}= \sigma\{\mathds{1}_{\{t+1 \le X \le k \wedge  (t+2)\}},  t+1 \le k \le x\}.
\]
Note now that we can limit the scope of the vector in \eqref{obsasproc} and in Lemma \ref{lemmaeinfach} to $x \in \{0,1, \ldots, \chi\}$ with $\chi \le s+G-1$. (For instance $\chi=5$  for the two bottom histories in Figure \ref{model}.)

\section{Representing the estimator as martingale} \label{mortalran}

In observational data, as of Section \ref{sec11}, the measurement of interest, $\mathbf{x}$, as well as any covariate, such as $t$, must not be considered differently. Hence the foundation year now is also a random variable, namely $T$ (see Figure \ref{model}). From a practical point of view, inference about enterprises founded hundreds of years in the past, appears to be of little use. We assume a short, but at least finite, support of $T$, without further specification. (Note that $T$ is uniformly distributed when assuming a homogeneous Poisson process for foundation \citep[][Lemma 2]{Doe}.) 

\vspace*{0.3cm}
\begin{enumerate}[label=\AnnANumm]
	\item \label{Grundmodellxt} Let $(X,T)'$ be define on $(\Omega, \mathcal{G}:=\mathcal{P}(\Omega), \tilde{P}_{\theta})$, with marginal probability density function of $X$ as in \eqref{Grundmodell0}, $T \in \{0,1, \ldots, G-1\}$ and $X$ be stochastically independent of $T$. 
\end{enumerate}
\vspace*{0.3cm}

Strictly, the notations in this section for random foundation and in the last section of fixed foundation are separated. For instance, $\tilde{P}_{\theta}$ stresses that it must not be $P_\theta$ of Section \ref{ltrcfixed} because $\Omega_X$ and $\Omega$ are unequal. It still is sometimes useful to compare models so that we use equal but redefined symbols for `very close' meanings. 
For the example in Figure \ref{domimagxt}, with $G=5$ of our data example, in $\Omega = \{\omega_1, \omega_2, \ldots \}$, we attribute $\omega_1, \ldots, \omega_5$ to $X=1$,  $\omega_6, \ldots, \omega_{10}$ to $X=2$, and so on. With $\mathbb{X}_1:=X^{-1}(1)=\{\omega_1, \ldots, \omega_5\}$, $\mathbb{X}_2:=\{\omega_6, \ldots, \omega_{10}\}$, and so on, $X$ is measurable with respect to $\mathcal{F}:= \sigma \{\mathbb{X}_1, \mathbb{X}_2, \ldots\} \subset \mathcal{G}$. (The symbol $\mathcal{F}$ is equal to Section \ref{ltrcfixed} because of the same meaning, but now must account for the random $T$.)

\setlength{\unitlength}{1.0cm}
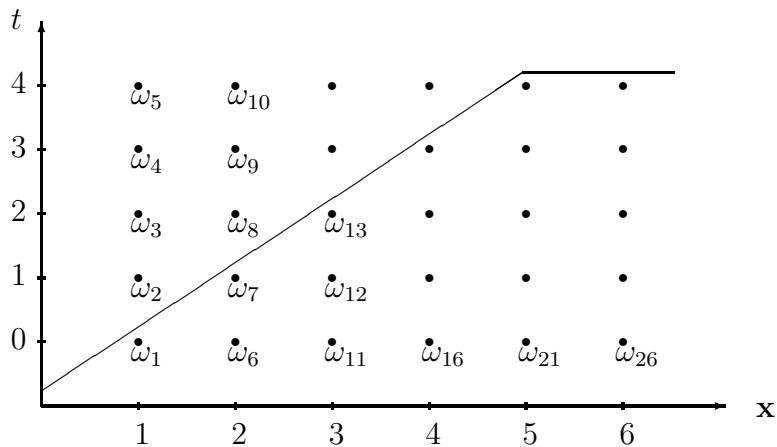
\begin{figure}[htb!] \centering
	\begin{picture}(12,6.0)
		\linethickness{0.3mm}
		\put(1.1,0.8){\vector(1,0){9}} \put(10.5,0.65){$\mathbf{x}$}
		\put(1.1,0.8){\vector(0,1){5.1}} \put(0.7,5.8){$t$}
		
		\put(1.05,1.65){\line(1,0){0.1}}
		\put(1.05,2.5){\line(1,0){0.1}}
		\put(1.05,3.35){\line(1,0){0.1}}
		\put(1.05,4.2){\line(1,0){0.1}}
		\put(1.05,5.05){\line(1,0){0.1}}
		
		\put(1.1,1.0){\line(3,2){6.33}}
		\put(7.42,5.22){\line(1,0){2}}

		\put(2.375,0.75){\line(0,1){0.1}}
		
		\put(3.65,0.75){\line(0,1){0.1}}
		
		\put(4.925,0.75){\line(0,1){0.1}}
		
		\put(6.2,0.75){\line(0,1){0.1}}
		
		\put(7.475,0.75){\line(0,1){0.1}}
		
		\put(8.75,0.75){\line(0,1){0.1}}
		
		
		\put(0.7,1.55){0}
		\put(0.7,2.4){1}
		\put(0.7,3.25){2}
		\put(0.7,4.1){3}
		\put(0.7,4.95){4}
		
		\put(2.325,0.3){1}
		\put(3.6,0.3){2}
		\put(4.875,0.3){3}
		\put(6.15,0.3){4}
		\put(7.425,0.3){5}
		\put(8.7,0.3){6}
		\linethickness{0.15mm}

		\put(2.375,1.65){\circle*{0.1}} \put(2.275,1.4){$\omega_1$}
		\put(2.375,2.5){\circle*{0.1}} \put(2.275,2.25){$\omega_2$}
		\put(2.375,3.35){\circle*{0.1}} \put(2.275,3.1){$\omega_3$}
		\put(2.375,4.2){\circle*{0.1}} \put(2.275,3.95){$\omega_4$}
		\put(2.375,5.05){\circle*{0.1}} \put(2.275,4.8){$\omega_5$}
		
		\put(3.65,1.65){\circle*{0.1}} \put(3.55,1.4){$\omega_6$}
		\put(3.65,2.5){\circle*{0.1}} \put(3.55,2.25){$\omega_7$}
		\put(3.65,3.35){\circle*{0.1}} \put(3.55,3.1){$\omega_8$}
		\put(3.65,4.2){\circle*{0.1}} \put(3.55,3.95){$\omega_9$}
		\put(3.65,5.05){\circle*{0.1}} \put(3.55,4.8){$\omega_{10}$}
		
		\put(4.825,1.4){$\omega_{11}$}
		\put(4.825,2.25){$\omega_{12}$}
		\put(4.825,3.1){$\omega_{13}$}
		\put(6.1,1.4){$\omega_{16}$}
		\put(7.375,1.4){$\omega_{21}$}
		\put(8.65,1.4){$\omega_{26}$}

		\put(4.925,1.65){\circle*{0.1}} 
		\put(4.925,2.5){\circle*{0.1}}
		\put(4.925,3.35){\circle*{0.1}}
		\put(4.925,4.2){\circle*{0.1}}
		\put(4.925,5.05){\circle*{0.1}}

		\put(6.2,1.65){\circle*{0.1}} 
		\put(6.2,2.5){\circle*{0.1}}
		\put(6.2,3.35){\circle*{0.1}}
		\put(6.2,4.2){\circle*{0.1}}
		\put(6.2,5.05){\circle*{0.1}}
		
		\put(7.475,1.65){\circle*{0.1}} 
		\put(7.475,2.5){\circle*{0.1}}
		\put(7.475,3.35){\circle*{0.1}}
		\put(7.475,4.2){\circle*{0.1}}
		\put(7.475,5.05){\circle*{0.1}}
		
		\put(8.75,1.65){\circle*{0.1}} 
		\put(8.75,2.5){\circle*{0.1}}
		\put(8.75,3.35){\circle*{0.1}}
		\put(8.75,4.2){\circle*{0.1}}
		\put(8.75,5.05){\circle*{0.1}}

	\end{picture}
	\caption{Year of foundation $T$ for the example of $G=5$ and lifespan $X$: domain and image of $(X,T)'(\omega)$, set below line $\mathbb{T}$ (linear up to and constant after $\mathbf{x}=5$) represents observable outcomes} \label{domimagxt}
\end{figure}

We define event $\mathbb{T}:=\{X \ge T+1\}$.  As is the case for the AFiD panel, we observe lifespan $X$ and foundation year $T$ in the case of $\mathbb{T}$, i.e. when closure has in 2017 not yet occurred. As in Section \ref{ltrcfixed}, we defer censoring to the end of the section. For the time being, we imagine the Table \ref{Data} to be `prolonged' toward the right until infinity. 
\vspace*{0.3cm}
\begin{enumerate}[label=\AnnANumm]
	\setcounter{enumi}{1}	
	\item \label{obscond} $(X,T)$ is observed if $X\ge T +1$.
\end{enumerate}
\vspace*{0.3cm}
 Neither measurement $X$ nor $T$ - and not even the enterprise at all - are observed when $X \le T$ (see again top path of Figure \ref{model}). 

\subsection{Decomposition for a latent filtration} \label{sec221}

The symbol $N$ is equal to Section \ref{ltrcfixed} because of the same definition. A filtration is, to know the year of foundation of an enterprise and whether, at the age of $x$, an enterprises is closed, and if when, i.e. $\mathcal{G}_x := \sigma\{ N(k), \mathds{1}_{\{T \le t\}}, 1 \le k \le x, 0 \le t \le G-1\}$. This filtration contains unobservable information about, latent experiments. For instance, for the top path of Figure \ref{model}, it stores for $x=3$ the answer that the enterprise had been founded in year $t=4$, i.e. in 2013, and that it had been closed in $x=2$. It will still become important to consider it. With respect to $\tilde{P}_{\theta}$, for the univariate (marginal) process it is (for $x \in \mathbb{N}$)
\begin{equation} \label{con2d1d}
E_{\theta}(\Delta N(x) \vert \mathcal{G}_{x-1})=\theta Y(x-1)=:\Delta A(x, \theta) \quad \text{and} \; A(0,\theta):=0.
\end{equation}
(Note that $E_{\theta} Z:= \int Z d \tilde{P}_{\theta}$ and there is not confusion with $E_{\theta}$ of Section \ref{ltrcfixed}.) The straight forward derivation of the conditional expectation in \eqref{con2d1d} uses that $T$ and $X$ are independent (by Assumption \ref{Grundmodellxt}). The notation $Y(x):= \mathds{1}_{\{X \ge x+1 \}}$ indicates whether an enterprise is not yet closed at the age of $x$. Hence the $(\mathcal{G}_x,\tilde{P}_{\theta})$-compensator of $N$ is, similar to \eqref{defmart}, $A(x, \theta)=\theta \sum_{k=1}^x  Y(k-1)$.

The $\sigma$-field $\mathcal{G}_x$ contains the information about the year of foundation $T$ so that $\{T > x\}$ can be answered,  i.e. one knows whether an enterprise in closed when the observation starts. Equivalently holds for its complement $\{T \le x\} \in \mathcal{G}_x$, i.e. T is a $\mathcal{G}_x$ stopping time. The stopping-time property will become important in the proof of the next theorem and is hence the merit of coding the foundation $T$ `backwards'. \citet[][Formula 35.20]{Billing} describes the $\sigma$-field of events that occur until $T$ as 
\begin{equation} \label{gtdef}
\mathcal{G}_T:=\{\mathbb{G} \subseteq \Omega: \mathbb{G} \cap \{T \le x\} \in \mathcal{G}_x \; \forall x \in \mathbb{N}\}.
\end{equation}

The intuition for $\leftindex_T{N}(x):=N(x) - N(x \wedge T)$ is unchanged as compared to that for the fixed $t$. And to be observably at risk, i.e. that closure occurs after 2017, is now formally $\leftindex_T{Y}(x):=Y(x) \mathds{1}_{\{x \ge T\}}=\mathds{1}_{\{X-1 \ge x \ge T\}}$. Consider the information at the age of $max(x,T)$, namely $\leftindex_T{\mathcal{G}}_x:= \mathcal{G}_x \vee \mathcal{G}_T$ \cite[for a definition see][page 20]{Chung}, which builds again a filtration.  With proof in Appendix \ref{proofcomptn}: 

\begin{theorem} \label{comptn}
The process $\leftindex_TN-\leftindex_TA$, with $\leftindex_TA(x, \theta):= \theta \sum_{k=1}^{x}  \leftindex_T{Y}(k-1)$, is a martingale with respect to the filtration 
 $\{ \leftindex_T {\mathcal{G}}_x: x \in \mathbb{N}_0\}$ and probability measure 	$\tilde{P}_{\theta}^{\mathbb{T}}(\mathbb{G}):=\tilde{P}_{\theta}(\mathbb{G} \cap \mathbb{T})/\tilde{P}_{\theta}(\mathbb{T})$ for $\mathbb{G} \in \mathcal{G}$.
\end{theorem}

Sometimes, the decomposition of a time-discrete counting process into a `trend', the compensator - and a `noise', the martingale, is called Doob decomposition \cite[][Theorem 9.3.2]{Chung}. (The much more involved Doob-Meyer decomposition for a time-continuous process \citep[see][Section 1.4]{flem1991} is fortunately not needed for the contingency Table \ref{Data}.)  The martingale property of $\leftindex_TN-\leftindex_TA$ implies that increments have expectation zero so that $\Delta \leftindex_T{A}(x,\theta)$ would be interpretable as probability of an Bernoulli experiment $\Delta \leftindex_T{N}(x)$.

Our aim is now to proceed as in Lemma \ref{lemmaeinfach} for the deterministic $t$, and derive a conditional likelihood contribution by interpreting the observable increments of $_TA$ (see Table \ref{Data} and later Table \ref{Datarow}, without the parameter) as probabilities of the observable Bernoulli experiments $\Delta _TN$ (see again Table \ref{Data}). Note that $\tilde{P}_{\theta}^{\mathbb{T}}$ is the probability measure $\tilde{P}_{\theta}$ restricted to $\mathbb{T}$, the set of observable outcomes.

Unfortunately, as mentioned at the beginning of the subsection $\mathcal{G}_x$  does not code the observable experiments, nor does  $\leftindex_T{\mathcal{G}}_x$, as we will now see. Hence, even though being conditional probabilities, they can not (yet) been used to build a conditional likelihood.

\subsection{Decomposition for the observed filtration}

The marginal information for the lifetime $X$ is stored in $\mathcal{F}$ (see comments after Assumption \ref{Grundmodellxt}). For the marginal information of the foundation year $T$ note first that foundation in 2017 is $T^{-1}(0)=\mathbb{T}_0:= \{\omega_1, \omega_6, \omega_{11}, \ldots \}$ (see Figure \ref{domimagxt}). Then $T^{-1}(1)= \mathbb{T}_1:= \{\omega_2, \omega_7, \ldots \}$ and so on for $\mathbb{T}_2$ until $\mathbb{T}_{G-1}$, so that $T$ is (minimal) measurable with respect to 
\begin{equation} \label{defg0}
	\mathcal{G}_0:= \sigma \{ \mathbb{T}_0, \ldots,  \mathbb{T}_{G-1}\}.
\end{equation}

 (Note that an infinite support for the distribution of $T$ would need to enumerate the elementary events of $\Omega$ differently; $\omega_6$ can become $\omega_3$, then $\omega_{11}$ will be $\omega_4$, and so an. Rational numbers are usually enumerated in that manner in their proof of countability.)

To know for an enterprise its year of foundation is less than to know for the enterprise everything that happens before one year before observation starts (i). And, together with the redefined $\mathcal{F}_x:=\sigma\{\mathbb{X}_0, \ldots, \mathbb{X}_x\}$ (equal in meaning to $\mathcal{F}_x$ of Section \ref{ltrcfixed}, but now accounting for the random $T$) the common information, until $x$, can be decomposed into marginal information (ii). 

\begin{lemma} \label{g0gt}
It is (i) $\mathcal{G}_0 \subseteq \mathcal{G}_T$ and (ii) $\mathcal{G}_x = \mathcal{G}_0 \vee \mathcal{F}_x$. 
\end{lemma}

\begin{proof}
For (i),  recall the definition \eqref{gtdef}. We aim at proving that each (composing) element of $\mathcal{G}_0$ is in  $\mathcal{G}_T$, i.e. that for every $y=0, \ldots, G-1$ it is $\mathbb{T}_y \cap \{T \le x\} \in \mathcal{G}_x \; \forall x$. We note that (see definition of $\mathcal{G}_x$  at the beginning of Section \ref{sec221})
$\mathcal{G}_x=\sigma \{\omega_1, \ldots \omega_{5x}, \mathbb{T}_0, \ldots, \mathbb{T}_4\}$ and distinguish cases (a) $x \le y$ and (b) $x > y$. Consider as examples (a) $y=3$, $x=2$ and (b)  $x=3$, $y=2$. For (a) it is $\mathbb{T}_3  \cap  \{T \le 2\} = \emptyset  \in \mathcal{G}_2$ because $\{\omega: T(\omega) \le 2\}  =  \mathbb{T}_0 \cup \mathbb{T}_1 \cup \mathbb{T}_2$ (\text{strictly it is $T(\omega)$ the second coordinate of $(X,T)'(\omega)$}). 
For (b) see that $\mathbb{T}_2  \cap  \{T \le 3\}  =  \mathbb{T}_2 \in \mathcal{G}_3$. The examples easily generalize. For (ii), note that $\mathcal{F}_x=\sigma \{\mathbb{X}_1, \ldots, \mathbb{X}_x\}=\sigma \{N(k); k=0, \ldots, x\}$ and $\mathcal{G}_0$ as in \eqref{defg0}. Therefore it is  $\mathcal{G}_0=\sigma\{\mathds{1}_{\{T \le k\}}; k=0, \ldots G-1\}$. \qed
\end{proof}

Using the definition and Lemma \ref{g0gt} it is 
\begin{equation} \label{zertgx}
	\leftindex_T{\mathcal{G}}_x= \mathcal{G}_x \vee \mathcal{G}_T=\mathcal{G}_T \vee \mathcal{G}_0 \vee \mathcal{F}_x = \mathcal{G}_T \vee \mathcal{F}_x.
\end{equation}

We see that $\leftindex_T{\mathcal{G}}_x$ is finer than $\mathcal{G}_T$ and - with its definition 
\eqref{gtdef} - all information for any enterprise up to age $T$ is available. This is not the case e.g. for the top path in Figure \ref{model}, where the information about its foundation date has been lost and also e.g. the information $\mathcal{F}_1$, whether the enterprise is closed or not after one year, is not an observed experiment. A coarser filtration is necessary to describe the observed history. 
The information starting from $T$ of the process $\leftindex_TN$ (until $x$) is given by   
\[
 \sigma\{\leftindex_TN(k): T < k \leq x\}:= \sigma\{\leftindex_TN(k\vee T): 1 \le k \leq x\}.
\]
This extends to the generation of a $\sigma$-field by two counting processes. The $\sigma$-field $\leftindex_T {\mathcal{F}}_x: = \sigma \{\leftindex_TN(k), \leftindex_TY(k): T \leq k \leq x\}$ is observed because $\leftindex_TN$ and $\leftindex_TY$ are observed. Whereas
$\leftindex_T{\mathcal{G}}_x$ contains the answer to any  possible question about an enterprise - for the present and the past - at the age of $x \vee T$, $\leftindex_T {\mathcal{F}}_x$ only contains answers to some questions about the enterprise at the same age. That means $\leftindex_T {\mathcal{F}}_x$ is coarser than $\leftindex_T {\mathcal{G}}_x$.

\begin{lemma} \label{sammelfilt}
	For any $x \in \mathbb{N}_0$ it is $\leftindex_T {\mathcal{F}}_x \subseteq 	\leftindex_T{\mathcal{G}}_x$. 
\end{lemma}

\begin{proof}
$N$ is adapted to $\{\mathcal{F}_x: x \in \mathbb{N}_0\}$ by definition and $Y$ is also adapted the $\{\mathcal{F}_x\}$ because $Y(x)=0 \Leftrightarrow N(x)=1$ and vice versa. Because of Lemma \ref{g0gt}(ii) and \eqref{zertgx} it is $\mathcal{F}_x \subset \mathcal{G}_x \subseteq \leftindex_T{\mathcal{G}}_x$, for all $x \in \mathbb{N}_0$, so that $N$ and $Y$ are also adapted to $\leftindex_T{\mathcal{G}}_x$. 
Furthermore follows by definition that $\leftindex_TN(x)=\mathds{1}_{\{T+1 \leq X \}} N(x)$.

Now $\{T+1 \leq X\}$ is measurable with respect to $\leftindex_T{\mathcal{G}}_x$ for all $x \in \mathbb{N}_0$ because it is an event that occurs until $T$ and is hence in $\mathcal{G}_T$ (see \eqref{gtdef} and comment shortly before) and finally it follows by \eqref{zertgx}. The process $\mathds{1}_{\{T+1 \leq X \}}$ is then adapted to $\leftindex_T{\mathcal{G}}_x$ (due to completion with the complement). Finally the product of two measurable (real-valued) functions is measurable \cite[see e.g.][Theorem 13.3]{Billing}.  

Furthermore $\{T\leq x\} \in \mathcal{G}_x \subseteq \leftindex_T{\mathcal{G}}_x$ because $T$ is a $\mathcal{G}_x$-stopping time and by  \eqref{zertgx}. Again $\mathds{1}_{\{T \leq x\}}$ is measurable with respect to $\leftindex_T{\mathcal{G}}_x$. So that $\leftindex_TY(x)= \mathds{1}_{\{T \le x\}}   Y(x)$ as product of two $\leftindex_T{\mathcal{G}}_x$-adapted processes is itself $\leftindex_T{\mathcal{G}}_x$-adapted.

 $_TY$ and $_TN$ are adapted to $\{\leftindex_T {\mathcal{G}}_x\}$, So both domains of $_TY(x)$ and  are in  $\leftindex_T {\mathcal{G}}_x$ (as well as their intersections and unions), whereas $\leftindex_T {\mathcal{F}}_x$ contains by definition the minimal set of domains of  $(_TN(x), _TY(x))$. \qed
\end{proof}

Now by the definition of the conditional expectation $E(Z\vert \mathcal{G})$ to satisfy $\int_G E(Z\vert \mathcal{G}) dP = \int_G Z d P \; \forall G \in \mathcal{G}$ \cite[see][Section 34]{Billing} it is
for $_TM:=_TN - \leftindex_T{A}$ and $\mu < \nu$:
\[
\int_{\mathbb{G}} E_{\theta}^{\mathbb{T}}( \leftindex_T{M}(\nu)\vert \leftindex_T{\mathcal{G}}_{\mu}) d \tilde{P}_{\theta}^{\mathbb{T}}= \int_{\mathbb{G}}  \leftindex_T{M}(\nu) d \tilde{P}_{\theta}^{\mathbb{T}} \quad \forall \mathbb{G} \in \leftindex_T{\mathcal{G}}_{\mu}
\]
by Theorem \ref{comptn} (recall definition $E_{\theta}^{\mathbb{T}}(Z)= \int Z d \tilde{P}_{\theta}^{\mathbb{T}}$.). Because of Lemma \ref{sammelfilt} the same is true $\forall \mathbb{G} \in \leftindex_T{\mathcal{F}}_{\mu}$, so that $_TM$ is also a martingale with respect to the observable filtration $\leftindex_T{\mathcal{F}}_x$. (The other conditions are clear by the finite-ness of $_TN$ and $_TA$.) This result is sometimes referred to as innovation theorem, and now
allows to interpret $\Delta \leftindex_TA(x)$ as observed probability of the Bernoulli experiment $\Delta \leftindex_TN(x)$. 

On the one hand $\sigma$-field $\leftindex_T {\mathcal{F}}_x$ does inevitably not contain information before $T$ because e.g. for an enterprise founded in 2013, i.e. with $T=4$, the question whether the enterprise is still active at the age of $x=1$ and closed at the age of $x=2$ could not be answered (see e.g. top path in Figure \ref{model}). On the other hand is deliberately  does not contain  information before $T$, even though it could have been answered for an enterprise surviving 2018 (see e.g. middle path in Figure \ref{model}). The use of doing so will become clear soon now when we aim to mimic \eqref{likecontr2} (second line) to derive a likelihood contribution.

\subsection{Likelihood contribution with condition $\mathds{1}_{\{T+1 \le X\}}$} \label{defclc}

For a contribution of enterprise $i$ to the conditional likelihood \citep[see again][Definition 7.2(ii)]{gourieroux1995}, we need to formulate $\leftindex_TL(\theta)$ for the both cases $\mathds{1}_{\{t+1 \le \mathbf{x}\}} \in \{0,1\}$ as in the proof of Lemma \ref{lemmaeinfach}.  If an enterprise is not observable $\mathds{1}_{\{t+1 \le \mathbf{x}\}}=0$, it is obviously $\leftindex_TL(\theta)=1$ because $\leftindex_TN \equiv \leftindex_TY \equiv 0$. Consider the opposite from now on. The Markov property again holds, and factorizing the pdf of $_TN$, as in Lemma \ref{lemmaeinfach}, now needs to take account of $\leftindex_TY$. Of course we can replace in $(\leftindex_TN, \leftindex_TY)$ the vector $\leftindex_TN$ by $\Delta \leftindex_TN$. With the notation 
\begin{equation}
	\leftindex_TJ:=\begin{pmatrix}
		\Delta \leftindex_TN(T+1) & \Delta \leftindex_TN(T+2) & \ldots & \Delta \leftindex_TN(\chi)  \\ 
		\leftindex_TY(T) &  \leftindex_TY(T+1) & \ldots & \leftindex_TY(\chi-1) 
	\end{pmatrix}, \label{matrixmarginal}
\end{equation}
the conditional likelihood contribution becomes the conditional pdf 
\[
f^{\leftindex_TJ \vert \mathds{1}_{\{T+1 \le X\}}}(\leftindex_Tj \vert \mathds{1}_{\{t+1 \le \mathbf{x}\}}),
\] 
where small letters (as in Lemma \ref{lemmaeinfach}) symbolize the outcome or realization (whereas the left lower capital index $T$ is a general symbol for the left-truncation design). We may order the elements of $_TJ$ differently, 
 so that the conditional pdf becomes
\begin{multline*}
f^{\leftindex_T\Delta N(\chi) \vert	\leftindex_T\Delta N(\chi-1), \ldots, 	\leftindex_T\Delta N(T+1), \leftindex_TY(\chi-1), \ldots, \leftindex_TY(T), \mathds{1}_{\{T+1 \le X\}}} \cdot \\
 \cdot f^{\leftindex_TY(\chi-1) \vert \leftindex_T\Delta N(\chi-1), \ldots, 	\leftindex_T\Delta N(T+1), \leftindex_TY(\chi-2), \ldots, \leftindex_TY(T), \mathds{1}_{\{T+1 \le X\}} } \cdot \\
 \cdot  f^{	\leftindex_T\Delta N(\chi-1) \vert	\leftindex_T\Delta N(\chi-2), \ldots, 	\leftindex_T\Delta N(T+1), \leftindex_TY(\chi-2), \ldots, \leftindex_TY(T), \mathds{1}_{\{T+1 \le X\}}}  \cdot  \\ 
 \ldots \cdot f^{\Delta \leftindex_TN(T+1) \vert \leftindex_TY(T), \mathds{1}_{\{T+1 \le X\}}} \cdot 
	f^{\leftindex_TY(T) \vert  \mathds{1}_{\{T+1 \le X\}}}.
\end{multline*}
(Note that the last formula could contribute to a full likelihood by the additional factor $f^{\mathds{1}_{\{T+1 \le X\}}}=\tilde{P}_{\theta}(T+1 \le X)$, the observation probability. However then, the specification of the distribution of $T$ would become necessary and the semi-parametric model must be abandoned.) We separate factors starting with $ _TY$ and $\Delta _TN$. For a factor starting with $_TY(x)$ and for $x \in \{T+1, \ldots, \chi-1\}$ it is:
\begin{multline*}
	f^{\leftindex_TY(x) \vert \leftindex_T\Delta N(x), \ldots, 	\leftindex_T\Delta N(T+1), \leftindex_TY(x-1), \ldots, \leftindex_TY(T), \mathds{1}_{\{T+1 \le X\}} } \\ \bigl(\leftindex_Ty(x) \vert \Delta 	\leftindex_Tn(x), \ldots, \Delta 	\leftindex_Tn(T+1),  \leftindex_Ty(x-1), \ldots, \leftindex_Ty(T), 1\bigr)\\
	 \stackrel{(i)}{=} f^{\leftindex_TY(x) \vert \leftindex_T\Delta N(x), \leftindex_TY(x-1) }   \bigl(\leftindex_Ty(x) \vert \Delta \leftindex_Tn(x), \leftindex_Ty(x-1) \bigr) \stackrel{(ii)}{=}1
\end{multline*}
Equality (i) is due to the Markovian property. For (ii) note that if $\leftindex_TY(x-1)=0$, then directly $\leftindex_TY(x)=0$. And furthermore results from $\leftindex_TY(x-1)=1$ and $\leftindex_TN(x)=0$ directly $\leftindex_TY(x)=1$. Also imply $\leftindex_TY(x-1)=1$ and $\leftindex_TN(x)=1$ that  $\leftindex_TY(x)=0$. The value of $\leftindex_TY(x)$ can hence be uniquely inferred from  $\leftindex_TY(x-1)$ and $\leftindex_TN(x)$ as well as $\leftindex_TY(x)$ is uniquely determined by the outcome of $\Delta \leftindex_TN(x)$ and $\leftindex_TY(x-1)$. (Note that $(\leftindex_Ty(x),\Delta \leftindex_Tn(x), \leftindex_Ty(x-1))$ is not defined over the entire $\{0,1\} \times \{0,1\} \times \{0,1\}$.) For the case $x=T$ it is
$f^{\leftindex_TY(T) \vert \mathds{1}_{\{T+1 \le X\}}}(\leftindex_Ty(T) \vert 1 )=1$.

For the factors starting with $\Delta _TN$, and for all $x \in \{T+1, \ldots, \chi\}$, the information about $\Delta \leftindex_TN(x-1), \ldots, \leftindex_T\Delta N(1)$ and  $\Delta \leftindex_TY(x-1), \ldots, \leftindex_T\Delta Y(1)$ is contained in $\leftindex_T{\mathcal{F}}_{x-1}$: 
\begin{multline*}
	f^{	\leftindex_T\Delta N(x) \vert	\leftindex_T\Delta N(x-1), \ldots, 	\leftindex_T\Delta N(T+1), \leftindex_TY(x-1), \ldots, \leftindex_TY(T),  \mathds{1}_{\{T+1 \le X\}} }  \\ \bigl(\Delta 	\leftindex_Tn(x) \vert \Delta 	\leftindex_Tn(x-1), \ldots, \Delta 	\leftindex_Tn(t+1),  \leftindex_Ty(x-1), \ldots, \leftindex_Ty(t), 1 \bigr) \\
	= \tilde{P}_{\theta}^{\mathbb{T} }\{\Delta \leftindex_TN(x)=1 \vert \leftindex_T{\mathcal{F}}_{x-1}\}^{\Delta \leftindex_Tn(x)} \cdot \tilde{P}_{\theta}^{\mathbb{T} }\{\Delta \leftindex_TN(x)=0 \vert \leftindex_T{\mathcal{F}}_{x-1}\}^{1- \Delta \leftindex_Tn(x)} \\
	= \tilde{P}_{\theta}^{\mathbb{T} }\{\Delta \leftindex_TN(x)=1 \vert \leftindex_T{\mathcal{F}}_{x-1}\}^{\Delta \leftindex_Tn(x)} \cdot \bigl(1- \tilde{P}_{\theta}^{\mathbb{T}}\{\Delta \leftindex_TN(x)=1 \vert \leftindex_T{\mathcal{F}}_{x-1}\}^{1- \Delta \leftindex_Tn(x)} \bigr)   \\
	= \bigl(\mathrm{E}_\theta^{\mathbb{T} }\{\Delta \leftindex_TN(x) \vert \leftindex_T{\mathcal{F}}_{x-1}\} \bigr)^{\Delta \leftindex_Tn(x)} \bigl(1- \mathrm{E}_\theta^{\mathbb{T} }\{\Delta \leftindex_TN(x) \vert \leftindex_T{\mathcal{F}}_{x-1}\}\bigr)^{1-\Delta \leftindex_Tn(x)}
\end{multline*}
 For the last equality see comment (II) shortly before \eqref{defmart}. 
The process $\leftindex_TA$ is by Lemma \ref{sammelfilt} the $(\leftindex_T{\mathcal{F}}_x,\tilde{P}_\theta^{\mathbb{T}})$-compensator of $\leftindex_TN$. Furthermore, note that, for an $x \in \{T+1, \ldots, \chi\}$, such that $\leftindex_TN(x)=1$, follows $\Delta \leftindex_TN(k)=1$ for some $k \in \{T+1, \ldots, x\}$. And $\Delta \leftindex_TN(x)=1$ implies that $\leftindex_TN(x)=1$ and  $\leftindex_TN(x-1)=0$.
By evaluatiom at the random variables $_TN$ and $_TY$, the conditional likelihood is:
\begin{eqnarray*}
	\leftindex_TL(\theta) & = & \prod_{x=T+1}^\chi \bigl\{ (1- \Delta \leftindex_T A(x, \theta) )^{1-  \Delta \leftindex_TN(x)} (\Delta \leftindex_TA(x, \theta))^{\Delta \leftindex_TN(x)} \bigr\} \nonumber \\ 
	& = & \Bigl[\prod_{x=T +1}^{X-1} (1-\theta)^1 \theta^0 \Bigr] \cdot (1-\theta)^0 \theta^1 \cdot \Bigl[\prod_{x \ge X} 1^1 0^0 \Bigr] = \theta (1-\theta)^{(X-1)-T}. \label{LikeTrunkiert}
\end{eqnarray*}
The second equality requires $\chi$ to be larger than $X$. Note further that $\Delta \leftindex_TN(x) =   N(x)- N(x \wedge T) - N(x-1)+N((x-1) \wedge T) = \mathds{1}_{\{T \le x-1\}} \mathds{1}_{\{x=X\}}$ and  $\Delta \leftindex_TA(x, \theta) =  \theta \mathds{1}_{\{T\le  x-1 \leq X-1\}}$.

\subsection{Adjustment for right-censoring} \label{adfrc}

We proceed as in Section \ref{ltrcfixed} and define a left-truncated and right-censored counting process as $\leftindex_TN^c(x)  = \sum_{k=1}^{x } C(k) \Delta \leftindex_TN(k) 
	= \mathds{1}_{\{T< X \leq x \wedge  (T+2)\}}$ 
with compensator ($	\leftindex_TY^c(x):=C(x+1)\leftindex_TY(x)=\mathds{1}_{\{x+1 \leq T+2\}} \mathds{1}_{\{x \geq T\}} \mathds{1}_{\{X \geq x+1\}} $ and redefined $C(x):=\mathds{1}_{\{x \le T+2\}}$, we suppress the filtration here, but note that with $T$, $T+2$ is also a stopping time):
	\begin{eqnarray*}
	\leftindex_TA^c(x, \theta) & = & \sum_{k=1}^{x} C(k) \Delta \leftindex_TA(k, \theta) 
	= \theta \sum_{k=1}^x \leftindex_TY^c(k-1) \\
	& = & \theta \sum_{k=1}^x \mathds{1}_{\{k \leq T+2\}} \mathds{1}_{\{k-1 \geq T\}} \mathds{1}_{\{X \geq k\}} 
	=  \theta  \sum_{k=1}^x  \mathds{1}_{\{T< k \leq X \wedge (T+2)\}} 
\end{eqnarray*}

So that finally (practically $\chi$ can be $T+2$, theoretically it is not considered here) $\leftindex_TL^c(\theta)=1$ if $X \le T$ and else:
\begin{eqnarray} \label{likcont}
 &  & \prod_{T+1}^{T+2}\bigl\{ (1- \Delta \leftindex_T A^c(x, \theta) )^{1-  \Delta \leftindex_TN^c(x)} (\Delta \leftindex_TA^c(x, \theta))^{\Delta \leftindex_TN^c(x)} \bigr\} \nonumber \\
	& = & \Bigl[\prod_{x= T +1}^{(X-1) \wedge (T+2)} (1-\theta)^1 \theta^0 \Bigr]  \Bigl[(1-\theta)^0 \theta^1 \Bigr]^{\mathds{1}_{\{X \leq T+2\}}} \Bigl[\prod_{x \ge ((X-1) \wedge (T+2))+1} 1^1 0^0 \Bigr]  \nonumber \\
	&= & \theta^{\mathds{1}_{\{X \leq T+2\}}} (1-\theta)^{((X-1) \wedge (T+2))-T} 
\end{eqnarray}
because
$\Delta \leftindex_TN^c(x)  =  \mathds{1}_{\{T-1 \le x\leq T+2\}} \mathds{1}_{\{x=X\}}$ and $\Delta \leftindex_TA^c(x, \theta)  =  \theta \mathds{1}_{\{T< x \leq X \wedge (T+2)\}}$.

\section{Asymptotic normality and standard error} \label{appsec4}

Our overall aim is an approximate confidence interval for $\theta$, and the first of two assumptions - a restriction for the true parameter - is standard, and practically without consequences. 
\vspace*{0.3cm}
\begin{enumerate}[label=\AnnANumm]
	\setcounter{enumi}{2}
	\item \label{pararaumeinfach} It is for the true parameter $\theta_0 \in \Theta:=[\varepsilon;1-\varepsilon]$ for some small $\varepsilon \in (0;\frac{1}{2})$.
\end{enumerate}
\vspace*{0.3cm}

As usual, we assume the data to be the truncated part of a simple randon sample (SRS) drawn from the population \cite[see e.g.][]{And0,efron1999, weiswied2021,topaweis2024}. 
\vspace*{0.3cm}
\begin{enumerate}[label=\AnnANumm]
	\addtocounter{enumi}{3}
	\item\label{A3:Ind} $(X_i,T_i)'$, $i=1,\ldots,n, n \in \mathbb{N}$, is an SRS, i.e. $\mbox{i.i.d.}$ random variables. 
\end{enumerate}
\vspace*{0.3cm}

We perform \citep[see again][Definition 7.2(ii)]{gourieroux1995} a maximization of the conditional likelihood  and attach for each enterprise $i$ (observed or not) the index to $\leftindex_TN^c$, $\leftindex_TY^c$ and finally to the conditional probability distribution function \eqref{likcont}. The probability distribution of 
\[
((_TN^c_1,\leftindex_TY^c_1)', \ldots, (_TN^c_n,\leftindex_TY^c_n)'),
\] 
conditioning on $(\mathds{1}_{\{T_1+1 \le X_1\}}, \ldots, \mathds{1}_{\{T_n+1 \le X_n\}})$, is now the product of the conditional pdf's due to the independence in Assumption \ref{A3:Ind}:
\begin{equation} \label{defclike}
\leftindex_TL^c(data \vert \theta):= \prod_{i=1}^n 	\leftindex_TL_i^c(\theta)
\end{equation}

Other studies (similar in design but not equal to that in Section \ref{sec11}), may not have an observation period of two (but $s$) years and the population may not encompass five (but $G$) years. The coming results are  given in the respective form. Note also that the observable support of the $X$ is the maximum of the $\chi_i$'s (from Section \ref{defclc}) and results in $s+G-1$ years. 

We now study the maximum of the logarithmic conditional likelihood \eqref{defclike}, i.e. the sum of the $\log \leftindex_TL_i^c(\theta)$'s  (including zeros for the unobserved enterprises). We define $\leftindex_TN_\cdot^c(x):=\sum_{i=1}^n \leftindex_TN_i^c(x)$ and $\leftindex_TY_{\cdot}^c(x):= \sum_{i=1}^n   \leftindex_TY_i^c(x)$.

A first representation of the ratio that estimates $\theta_0$, collected in Lemma \ref{3rep}, follows directly from equating the derivative $\log \leftindex_TL^c(data \vert \theta)$ (using the first line of \eqref{likcont}) and solving for $\theta$. The second representation in Lemma \ref{3rep}, useful for computation, results from: 
\begin{equation} \label{nyexp}
	\begin{split}
		\sum_{i=1}^n \sum_{x=1}^{s+G-1} \Delta \leftindex_TN_i^c(x) & =  \sum_{i=1}^n  \sum_{x=1}^{s+G-1} \mathds{1}_{\{T_i < x \leq T_i+s\}} \mathds{1}_{\{x
			=X_i\}} =\sum_{i=1}^n  \mathds{1}_{\{T_i < X_i \leq T_i+s\}} \\
		\sum_{k=1}^x \leftindex_TY_i^c(k-1) & = \mathds{1}_{\{T_i< (X_i \wedge x))\}} \bigl[(X_i \wedge (T_i+s) \wedge x)-T_i \bigr] \\ 		
		\sum_{i=1}^n \sum_{x=1}^{s+G-1} \leftindex_TY_i^c(x-1) & = \sum_{i=1}^n  \mathds{1}_{\{T_i< X_i\}} \bigl[(X_i \wedge (T_i+s))-T_i \bigr]
	\end{split}
\end{equation}
(The second line here follows from the definition in Section \ref{adfrc}.)
The third representation, as a martingale transform, allows the asymptotic analysis. 
\begin{lemma} \label{3rep}
By maximization of the conditional likelihood \eqref{defclike} results for the point estimator $\hat{\theta}_n$ and is distance to the true parameter $\theta_0$:  
\begin{eqnarray*}
	\hat{\theta}_n & = &  \frac{\sum_{x=1}^{s+G-1} \Delta \leftindex_TN_{\cdot}^c(x)}{ \sum_{x=1}^{s+G-1} \leftindex_TY_{\cdot}^c(x-1)}  
	 =  \frac{\sum_{i=1}^n \mathds{1}_{\{T_i \le X_i-1 \leq T_i+1\}}}{\sum_{i=1}^n \mathds{1}_{\{T_i \le X_i-1\}} \bigl[   (X_i \wedge (T_i+2))-T_i \bigr]}  \\
		\sqrt{n}(\hat{\theta}_n & - & \theta_0)  =    \frac{1}{\sqrt{n}} \frac{\sum_{x=1}^{s+G-1} \bigl(\Delta \leftindex_TN_\cdot^c(x) - \theta_0 \leftindex_TY_\cdot^c(x-1) \bigr)}{\frac{1}{n} \sum_{x=1}^{s+G-1} \leftindex_TY_\cdot^c(x-1)}
\end{eqnarray*}
\end{lemma}

As an early example of the analysis from truncated data, \cite{heckman1976} did not distinguish between the enumeration of units in the (latent) sample and of the observations. (The observed units are thought to be sorted to the beginning of the sample.) For clearness we now enumerate the observations with a new index $j$. Also let us denote the realized lifespans and truncation ages of observed enterprises as $x^{\text{obs}}$ and $t^{\text{obs}}$, respectively, and that part of the lifespan in the study period as $d^{\text{obs}}_j:=x^{\text{obs}}_j - t^{\text{obs}}_j$. The number of (realized) observations, $m:=\sum_{i=1}^n \mathds{1}_{\{\mathbf{x}_i -1 \ge t_i\}}$, can now be split into $m_{uncens}:=\sum_{i=1}^n \mathds{1}_{\{t_i \le \mathbf{x}_i -1 \le t_i+2\}}$ uncensored and $m_{cens}:=\sum_{i=1}^n \mathds{1}_{\{\mathbf{x}_i+1-1 \ge t_i+s+1 \}}$ censored observation. A practical version of the point estimate in Lemma \ref{3rep}, omitting the subscript $n$ to avoid confusion, now is 
\begin{equation} \label{thetaprac}
	\hat{\theta}= \frac{m_{uncens}}{\sum_{j=1}^{m_{uncens}} d^{\text{obs}}_j + s m_{cens}}.
\end{equation}

 Coming back to theory, consider the numerator in the last expression of Lemma \ref{3rep} and define 
$\leftindex_TM_\cdot^c(x):= \leftindex_TN_\cdot^c(x) - \theta_0 \sum_{k=1}^{x-1}\leftindex_TY_\cdot^c(k)$. The martingale differences
\[
\leftindex^{PB}Y_{nx} := \frac{1}{\sqrt{n}} \frac{1}{\theta_0(1-\theta_0)} \Delta \leftindex_TM_\cdot^c(x)
\]
have conditional variances $\leftindex^{PB}{\sigma}^2_{nx}:= E_{\theta_0}^{\mathbb{T}}(\leftindex^{PB}Y^2_{nx} \vert \leftindex_T {\mathcal{F}}_{x-1})= \Delta \langle\leftindex_TV_{\cdot}^c \rangle(x)$. (The definition for the last equality will appear in the proof, and note that the expectation is unconditional.) In the notation of \citet[][Formula 35.35]{Billing} (and with proof in Appendix \ref{proofbedingung1}) it is in probability by the law of large numbers  
\begin{equation} \label{bedingung1}
\sum_{x=1}^{s+G-1} \leftindex^{PB}{\sigma}^2_{nx}  \stackrel{n \to \infty}{\longrightarrow}  \leftindex^{PB}{\sigma}^2 := (\theta_0(1-\theta_0))^{-1} \sum_{x=1}^{s+G-1} \mathrm{E}_{\theta_0}\{\leftindex_TY_1^c(x-1)\}.
\end{equation}
Note on the one hand, that Assumption \ref{pararaumeinfach} is needed is this step and on the other hand that $E_{\theta_0}\{\leftindex_TY_1^c(x-1)\}$ can be (infeasibly) estimated by $n^{-1} \sum_{i=1}^n \leftindex_TY_i^c(x-1)$. 

The main obstacle is that the Lindeberg-L\'{e}vy Condition of \citet[][Theorem 35.12]{Billing} is too strong. It can not be proven for the model at hand. A weaker version was derived in the 1970s by R. Rebolledo \cite[][Theorem 5.1]{the-statis:2002} and is easy to be verified
(\cite[][Formula 35.35]{Billing}, namely \eqref{bedingung1}, is exactly \citet[][Theorem 5.1, Condition (a)]{the-statis:2002}.) 
For their Condition (c) note that by Assumption \ref{pararaumeinfach}
$(\sqrt{n}\theta_0(1-\theta_0))^{-1} \stackrel{n \to \infty}{\longrightarrow} 0$
and hence 
\begin{align*}
	\langle \leftindex_TV_{\cdot, \epsilon}^c \rangle(x) &:=  \frac{1}{\theta_0 (1-\theta_0)}  \sum_{k=1}^x \mathds{1}_{\{\vert(\sqrt{n}\theta_0(1-\theta_0))^{-1}\vert > \epsilon \}} \frac{1}{n}\sum_{i=1}^n  \leftindex_TY_i^c(k-1) \\
	&\xrightarrow{p} \frac{1}{\theta_0 (1-\theta_0)}  \sum_{k=1}^x 0 \cdot \mathrm{E}_{\theta_0}\bigl\{\leftindex_TY_i^c(k-1)\bigr\} =0 
\end{align*} 

We can now state consistency and asymptotic normality.
\begin{theorem} \label{martiKT}
	Under Assumptions \ref{Grundmodellxt}-\ref{A3:Ind} with $s \le 2$ and $Z  \sim N(0,1)$ it is
$\sum_{x=1}^{s+G-1} \leftindex^{PB}Y_{nx} \Rightarrow \leftindex^{PB}\sigma Z$, for $n \to \infty$.
\end{theorem}

Now written in long, by Theorem \ref{martiKT} it is	
$\frac{1}{\sqrt{n}} \sum_{x=1}^{s+G-1} \bigl(\Delta \leftindex_TN_\cdot^c(x) - \theta_0 \leftindex_TY_\cdot^c(x-1) \bigr) \Rightarrow  (\theta_0(1-\theta_0) \sum_{x=1}^{s+G-1} \mathrm{E}_{\theta_0} \{\leftindex_TY_i^c(x-1)\})^{0.5} \cdot Z$.
Due to $\mathrm{E}_{\theta_0} \{\leftindex_TY_i^c(x-1)\}$ being constant, Slutsky's lemma  \cite[][Lemma 2.8]{vanderVaart} yields: 	
\begin{multline*}
	\frac{1}{\sqrt{n}} \frac{\sum_{x=1}^{s+G-1} \bigl(\Delta \leftindex_TN_\cdot^c(x) - \theta_0 \leftindex_TY_\cdot^c(x-1) \bigr)}{\frac{1}{n} \sum_{x=1}^{s+G-1} \leftindex_TY_\cdot^c(x-1)} \\ \Rightarrow \frac{\sqrt{\theta_0(1-\theta_0) \sum_{x=1}^{s+G-1} \mathrm{E}_{\theta_0} \{\leftindex_TY_i^c(x-1)\}}}{ \sum_{x=1}^{s+G-1} \mathrm{E}_{\theta_0} \{\leftindex_TY_i^c(x-1)\}} \cdot Z
\end{multline*}
Due to Lemma \ref{3rep}, the upper expression equals $\sqrt{n}(\hat{\theta}_n - \theta_0)$. The answer to the second question of Section \ref{introduc} is now the following. 
\begin{corollary} \label{corfinal}
Under Assumptions \ref{Grundmodellxt}-\ref{A3:Ind} with $s \ge 2$ and $Z  \sim N(0,1)$, 
it is
$\sqrt{n}(\hat{\theta}_n - \theta_0) \Rightarrow \leftindex^{PB}\sigma^{-1} Z$, for $n \to \infty$.
\end{corollary}

Note that even though $n$ is not observable in the estimation of $\leftindex^{PB}\sigma^2$,  
by
$Var(\sqrt{n} \hat{\theta}_n) =\leftindex^{PB}\sigma^{-2}$  it is 
$\widehat{Var(\hat{\theta}_n)}= n^{-1} \widehat{\leftindex^{PB}\sigma^{-2}}$ (with  $\hat{\theta}_n$ inserted for $\theta_0$) and $n$ cancels. Resuming the idea of a practice formula, for the square standard error and using the last equality in \eqref{nyexp} and  inserting \eqref{thetaprac} we have
\begin{equation} \label{seprac}
	\begin{split}
\widehat{Var(\hat{\theta})} & \approx  \frac{\hat{\theta}(1-\hat{\theta})}{ (\sum_{j=1}^{m_{uncens}} d^{\text{obs}}_j + s m_{cens})}  \\
     & =    \frac{m_{uncens}(\sum_{j=1}^{m_{uncens}} d^{\text{obs}}_j + s m_{cens})-m_{uncens}^2}{(\sum_{j=1}^{m_{uncens}} d^{\text{obs}}_j + s m_{cens})^3}. 
     \end{split}
\end{equation}

For the data described in Section \ref{sec11}, we now derive a confidence interval for the parameter $\theta$ of the geometric distribution modeling enterprise survival in Germany, and finally for the enterprise life expectancy.

\section{Result for AFiD panel, comparisons, simulations} \label{case1}

The sufficient statistics had already been given in Table \ref{Data}.

\subsection{Mapping of symbols and adequacy of assumptions} 

Specifically we have an observation period of `width' $s=2$ years, namely the years 2028 and 2019. We  consider as population all enterprises founded in the years 2013--2017, i.e. $G=5$. Table \ref{Datarow} gives a view on the raw data that easier shows the mapping of measuring foundation by calendar year to the age one year before observation starts, $T$ (see Figure \ref{model}).

\begin{table}[h]
			\caption{Counts of Table \ref{Data}, stratified by year of foundation.} \label{Datarow} 	
	\begin{center}
		\begin{tabular}{crrr}
			\hline \hline 
			& \multicolumn{3}{c}{Year of closure}    \\
			Foundation year &      2018  ($d^{obs}_j=1$)&      2019 ($d^{obs}_j=2$)&                2020 or later \\ \hline \hline
			2013 ($t=4$)        & 37,016 & 17,295 &  191,803 \\
			2014 ($t=3$)        & 42,272 & 20,305 &  200,411 \\
			2015 ($t=2$)       & 35,588 & 23,353 &  209,649 \\
			2016 ($t=1$)        & 34,579 & 27,464 &  278,223 \\
			2017 ($t=0$)       & 18,687 & 18,633 &  292,566  \\ \hline \hline
		\end{tabular}
	\end{center}
\end{table}

Assumption \ref{Grundmodellxt} requires that the closure probability does not vary with age. We know of no empirical evidence for the contrary. Fortunately, we do not need to argue about the foundation distribution, because the analysis is nonparametric in that direction. First (weak) indications against the independence assumption of $X$ and $T$ for business demography exist \citep{jacobo2024,topaweis2024}. (For human demography independence must be rejected, life expectation clearly increases over time \cite[see e.g.][]{vaupel1970}.) Assumption \ref{obscond} is strictly obeyed. Assumption \ref{pararaumeinfach} is not critical.  With respect to the independence Assumption \ref{A3:Ind} between enterprises, it must be admitted that some dependence is conceivable because lifespans overlap considerably. As attempt to cope with dependence, one may note that  Rebolledo's theorem does not explicitly require the enterprises, or their histories, to be independent. However, his assumption `$Y(t)/n \to \tilde{P}_{\theta}(X \le t)$', necessary for independence, to our knowledge is not fulfilled by important dependence models.

\subsection{Point and interval estimate}

 The point estimate, as given by Lemma \ref{3rep}, and the standard error, as given by Corollary \ref{corfinal}, are calculated in Table \ref{statmort}.
\begin{table}[h]
	\centering
	\caption{Statistics (as counts (No.) or in years) from Table \ref{Data} for point estimate \eqref{thetaprac} of the one-year closure probability and its standard error (SE) by \eqref{seprac}}
	\label{statmort} \smallskip
	{\footnotesize
		\begin{tabular}{ccccc}
			\hline
			\noalign{\smallskip}
			No. observ. & No. closures & `Time at risk' (years) & Point & SE \\ 
			{\em m} & {\em m}$_{uncens}$ & $\sum_{j=1}^{m_{uncens}} d_j^{obs} + 2 m_{cens} $  & $\hat{\theta}$ & $\sqrt{\widehat{Var(\hat{\theta})}}$ \\	\hline	
			\noalign{\smallskip}
			1,447,814 & 275,162  & 1$\cdot$168,112 + 2$\cdot$107,050+ 2$\cdot$1,172,652   & 0.1009 & $1.8 \cdot 10^{-4}$ \\  \hline
			\noalign{\smallskip}       
	\end{tabular}}
\end{table}
Together with the asymptotic normality due to Corollary \ref{corfinal}, an approximate confidence interval for $\theta$ (at level $1-\alpha=95\%$) results in $[0.1005;0.1013]$. (Our interval for the German enterprise life expectancy is hence $[9.88;9.95]$ so that the standard error for the expecation is $0.036$ and width of the confidence interval are two months.)

It should be added that Table \ref{Datarow} supplies, for each enterprise, the year of foundation. 
With that information, an age-inhomogeneous, i.e. $\mathbf{x}$-specific model, for survival is easily possible with changes starting in Formula \ref{conprobs}. Recall that the rather short support 2013 -- 2017 for foundation ($T$) is easily generalized. 

\subsection{Results by others} 

We relate now our estimation to that of others on enterprise demography, for different countries and/or times. For the comparisons one must cautiously keep in mind that imprecisely matching definitions of foundation and closure hinder the comparison \cite[see e.g.][Figure 1]{cochran1981}. \cite{bruederl1992} report, from their earlier German study of 1,621 enterprises, a 25\%-quantile of 2 years for the closure distribution and the 37\%-quantile as 5 years. The geometric distribution (Assumption \ref{Grundmodellxt}) implies for our data that 2 years is the 19\%-quantile and 5 years the 41\%-quantile, which are quite close findings. A study for Japan, also for earlier cohorts, reports 8 years as 4.8\%-quantile \citep{Honjo2000}, which is  much lower in closure intensity, compared to Germany. Our model implies 8 years as 57\%-quantile. However, for the more recent Japanese cohorts 2003 to 2013, the same author states that closures have double \citep{kato2023}. For Portugal, a median survival of 4 years (and a 20\%-quantile of 1 year) have been reported \citep{Mata1994} and imply higher closure intensity, compared to Germany. For Pakistan, a 80\%-quantile of 4 years is reported \citep{salim2016} and, compared to our 80\%-quantile of 15 years, indicates much faster closure. Also faster, by not by that much, is a median of 5 years until failure in the USA \cite{cochran1981}, whereas Germany (now) has a 7-years median.

\subsection{Simulation}  \label{simu}

In a short simulation, the consistency property of $\hat{\theta}_n$ from Lemma \ref{3rep} under the conditions of the application is studied by means of the (simulated) mean square error (MSE) $\frac{1}{K_{sim}} \sum_{K=1}^{K_{sim}} (\hat{\theta}^K_n - \theta_0)^2 $ over $K_{sim}=1,000$ simulation loops. Table \ref{nvar} confirms that for the application, with $n > 1 \; mio$, the MSE is small.  
\begin{table}[h]
	\caption{Mean square error for point estimate of Lemma \ref{3rep}, under conditions of this application in Section \ref{case1}, including $G=5$, $s=2$ and $\theta_0=0.1$, for some latent sample sizes $n$} \label{nvar}	
	\begin{center}
		\begin{tabular}{crrrrr} \hline 
			n&      100 &      1,000 &                10,000 & 100,000 &1,000,000\\ \hline\hline
			$\mathrm{MSE}$ & $8.57 \cdot 10^{-6}$ & $8.2 \cdot 10^{-8}$ & $6.96 \cdot 10^{-10}$& $7.39 \cdot 10^{-12}$& $8.11 \cdot 10^{-14}$\\ \hline \hline 
		\end{tabular}
	\end{center}
\end{table}

\section{Conclusion}
\label{sec4}

Quite commonly, event history data are time-discrete \citep[see e.g.][]{koopmann2008,stroweis2015,weissbachm2021effect}. Opposed to the time-continuous model \citep[see e.g.][]{flem1991}, the discrete time renders topics such as the Doob-Meyer decomposition, the product integral, continuity considerations for filtrations and the definition of a predictable $\sigma$-field abundant. An analysis can become accessible to a wider audience and gives room for further consideration. For instance for observational data, such as the AFiD panel here, Assumption \ref{A3:Ind} is violated \cite[see][Items 9, 12(c) and 12(d)]{strobe}. Luckily, the assumption of the latent measurements to be a simple random sample is not strong in the sense that the latent measurements considered as population, similarly for least squares considerations, maximizing the likelihood still results in the true parameter $\theta_0$. This is because such maximimum minimized - asymptotically - the Kulback-Leibler loss, and identification of the parameter in a geometric distribution is easily shown. However we like to admit that the interpretation of the standard error then is more subtle. Another source for improving the analysis, might be knowledge of a parametric distribution for $T$, e.g. the foundation process. The point estimate will not be altered, but the parametric analysis then only requires a central, and not a martingale, limit theorem and the use of M-estimation theory \cite[as in][Chapter 5]{vanderVaart} results in a smaller standard error than for semi-parametric modelling. 

We like to point out two minor technical aspects. First, as $n$ increases, the bounds to the counting processes (compensated or not) may increase without limits, so that the resulting martingale is not global anymore, but only local. The respective theory exists, but we suppressed it here. The second minor comment is that using a random truncation design, i.e. $T$ instead of $t$, is not only the the correct design for observational data, but has another advantage. The dominating measure $\tilde{P}_{\theta}^{\mathbb{T}}$ in Theorem \ref{comptn} depends on $T$, and would dependent on $t_i$ for fixed truncation. The product property of a joint probability density of independent random variables seems to require equal dominating measures of the marginal distributions. Random truncation resolves such in-ambiguity.

Interesting to us is, finally, that the work is a contribution to the statistical field of forecasting, because $\theta$ predicts survival for enterprises alive at the end the observation period.

\vspace*{0.7cm}
\noindent \textbf{Declarations}:
The authors declare that they have no conflict of interest.

\vspace*{0.3cm}
\textbf{Acknowledgment}:
We thank Dr. Claudia Gregor-Lawrenz (BaFin, Bonn) for continuous support on the topic, e.g. by contributing earlier data in the LTRC-design. We thank Dr. Florian K\"ohler (Destatis, Hannover) for supply of the data, Dr. Wolfram Lohse (Görg, Hamburg) for support in the data acquisition process, Simon Rommelspacher (Destatis, Wiesbaden) for advise on the measurement definitions. The financial support from the Deutsche Forschungsgemeinschaft (DFG) of R. Wei\ss bach is gratefully acknowledged (Grant 386913674 ``Multi-state, multi-time, multi-level analysis of health-related demographic events: Statistical aspects and applications'').


\begin{thebibliography}{}
	\providecommand{\doi}[1]{\url{https://doi.org/#1}}
	
	\bibitem[\protect\citeauthoryear{Andersen, Borgan, Gill, and Keiding}{Andersen
		et~al.}{1988}]{And0}
	Andersen, P., {\O}.~Borgan, R.~Gill, and N.~Keiding. 1988.
	\newblock Censoring, truncation and filtering in statistical models based on
	counting processes, In {\em Statistical inference from stochastic processes},
	ed. Prabhu, N.U., Volume~80,  19--60. Center for Mathematics and Computer
	Science, Amsterdam.
	
	\bibitem[\protect\citeauthoryear{Andersen, Borgan, Gill, and Keiding}{Andersen
		et~al.}{1993}]{And}
	Andersen, P., {\O}.~Borgan, R.~Gill, and N.~Keiding. 1993.
	\newblock {\em Statistical Models Based on Counting Processes}.
	\newblock Springer, New York.
	
	\bibitem[\protect\citeauthoryear{Audretsch and Mahmood}{Audretsch and
		Mahmood}{1995}]{Audretsch1995}
	Audretsch, D.B. and T.~Mahmood. 1995.
	\newblock New firm survival: New results using a hazard function.
	\newblock {\em Review of Economics and Statistics\/}~77: 97--103.
	
	\bibitem[\protect\citeauthoryear{Billingsley}{Billingsley}{2012}]{Billing}
	Billingsley, P. 2012.
	\newblock {\em Probability and Measure\/} (4${th}$ ed.).
	\newblock Wiley, New York.
	
	\bibitem[\protect\citeauthoryear{Brüderl, Preisendörfer, and
		Ziegler}{Brüderl et~al.}{1992}]{bruederl1992}
	Brüderl, J., P.~Preisendörfer, and R.~Ziegler. 1992.
	\newblock Survival chances of newly founded business organizations.
	\newblock {\em American Sociological Review\/}~57: 227--242.
	
	\bibitem[\protect\citeauthoryear{Chung}{Chung}{2001}]{Chung}
	Chung, K. 2001.
	\newblock {\em A Course in Probability Theory\/} (3$^{rd}$ ed.).
	\newblock Academic Press, San Diego.
	
	\bibitem[\protect\citeauthoryear{Cochran}{Cochran}{1981}]{cochran1981}
	Cochran, A. 1981.
	\newblock Small business mortality rates: A rewiew of the literature.
	\newblock {\em Journal of Small Business Economy\/}~19: 50--59.
	
	\bibitem[\protect\citeauthoryear{de~Uña-Álvarez, Martínez-Senra,
		Otero-Giráldez, and Quintás}{de~Uña-Álvarez et~al.}{2024}]{jacobo2024}
	de~Uña-Álvarez, J., A.~Martínez-Senra, M.~Otero-Giráldez, and M.~Quintás.
	2024.
	\newblock Cox regression with doubly truncated responses and timedependent
	covariates: The impact of innovation on firm survival.
	\newblock {\em Journal of Applied Statistics\/}~51: 780--792.
	
	\bibitem[\protect\citeauthoryear{Doob}{Doob}{1953}]{doob1953}
	Doob, J. 1953.
	\newblock {\em Stochastic Processes}.
	\newblock John Wiley \& Sons, New York.
	
	\bibitem[\protect\citeauthoryear{{D\"orre}}{{D\"orre}}{2020}]{Doe}
	{D\"orre}, A. 2020.
	\newblock Bayesian estimation of a lifetime distribution under double
	truncation caused by time-restricted data collection.
	\newblock {\em Statistical Papers\/}~61: 945--965.
	
	\bibitem[\protect\citeauthoryear{Efron and Petrosian}{Efron and
		Petrosian}{1999}]{efron1999}
	Efron, B. and V.~Petrosian. 1999.
	\newblock Nonparametric methods for doubly truncated data.
	\newblock {\em Journal of the American Statistical Association\/}~94: 824--834.
	
	\bibitem[\protect\citeauthoryear{Feller}{Feller}{1968}]{Fel1}
	Feller, W. 1968.
	\newblock {\em An Introduction to Probability Theory and Its Applications\/}
	(3$^{rd}$ ed.), Volume~1.
	\newblock Wiley, New York.
	
	\bibitem[\protect\citeauthoryear{Fleming and Harrington}{Fleming and
		Harrington}{1991}]{flem1991}
	Fleming, T.R. and D.~Harrington. 1991.
	\newblock {\em Counting Processes and Survival Analysis}.
	\newblock Wiley, Hoboken.
	
	\bibitem[\protect\citeauthoryear{Gouri{\'e}roux and Monfort}{Gouri{\'e}roux and
		Monfort}{1995}]{gourieroux1995}
	Gouri{\'e}roux, C. and A.~Monfort. 1995.
	\newblock {\em Statistics and Econometric Models}, Volume~1.
	\newblock Cambridge University Press, Cambridge.
	
	\bibitem[\protect\citeauthoryear{Heckman}{Heckman}{1976}]{heckman1976}
	Heckman, J. 1976.
	\newblock The common structure of statistical models of truncation, sample
	selection and limited dependent variables and a simple estimator for such
	models.
	\newblock {\em Annals of Economic and Social Measurement\/}~5: 475--492.
	
	\bibitem[\protect\citeauthoryear{Hernán, Sauer, Hernández-Díaz, Platt, and
		Shrier}{Hernán et~al.}{2016}]{HERNAN201670}
	Hernán, M., B.~Sauer, S.~Hernández-Díaz, R.~Platt, and I.~Shrier. 2016.
	\newblock Specifying a target trial prevents immortal time bias and other
	self-inflicted injuries in observational analyses.
	\newblock {\em Journal of Clinical Epidemiology\/}~79: 70--75.
	
	\bibitem[\protect\citeauthoryear{Honjo}{Honjo}{2000}]{Honjo2000}
	Honjo, Y. 2000.
	\newblock Business failure of new firms: An empirical analysis using a
	multiplicative hazards model.
	\newblock {\em International Journal of Industrial Organization\/}~18: 557--574.
	
	\bibitem[\protect\citeauthoryear{Johnson, Kemp, and Kotz}{Johnson
		et~al.}{2005}]{Jo0}
	Johnson, N., A.~Kemp, and S.~Kotz. 2005.
	\newblock {\em Univariate Discrete Distributions\/} (3$^{rd}$ ed.).
	\newblock John Wiley \& Sons, New York.
	
	\bibitem[\protect\citeauthoryear{Kalbfleisch and Prentice}{Kalbfleisch and
		Prentice}{2002}]{the-statis:2002}
	Kalbfleisch, J. and R.~Prentice. 2002.
	\newblock {\em The Statistical Analysis of Failure Time Data\/} (2$^{nd}$ ed.).
	\newblock John Wiley \& Sons, New York.
	
	\bibitem[\protect\citeauthoryear{Kato, Onishi, and Honjo}{Kato
		et~al.}{2022}]{kato2023}
	Kato, A., K.~Onishi, and Y.~Honjo. 2022.
	\newblock Does patenting always help new firm survival? understanding
	heterogeneity among exit routes.
	\newblock {\em Small Business Economy\/}~59: 449--475.
	
	\bibitem[\protect\citeauthoryear{Klenke}{Klenke}{2020}]{klenke}
	Klenke, A. 2020.
	\newblock {\em Wahrscheinlichkeitstheorie\/} (4$^{th}$ ed.).
	\newblock Springer, Berlin.
	
	\bibitem[\protect\citeauthoryear{Koopmann and Lucas}{Koopmann and
		Lucas}{2008}]{koopmann2008}
	Koopmann, S. and A.~Lucas. 2008.
	\newblock A non-gaussian panel time series model for estimating and decomposing
	default risk.
	\newblock {\em Journal of Business and Economic Statistics\/}~26: 510--525.
	
	\bibitem[\protect\citeauthoryear{Mata and Portugal}{Mata and
		Portugal}{1994}]{Mata1994}
	Mata, J. and P.~Portugal. 1994.
	\newblock Life duration of new firms.
	\newblock {\em Journal of Industrial Economics\/}~42: 227--245.
	
	\bibitem[\protect\citeauthoryear{Michimae, Emura, Miyamoto, and Kishi}{Michimae
		et~al.}{2024}]{michemura24}
	Michimae, H., T.~Emura, A.~Miyamoto, and K.~Kishi. 2024.
	\newblock Bayesian parametric estimation based on left-truncated competing
	risks data under bivariate clayton copula models.
	\newblock {\em Journal of Applied Statistics\/}~51: 2690--2708.
	
	\bibitem[\protect\citeauthoryear{Pittiglio}{Pittiglio}{2023}]{pittiglio2023}
	Pittiglio, R. 2023.
	\newblock Counterfeiting and firm survival. do international trade activities
	matter?
	\newblock {\em International Business Review\/}~32: 102145.
	
	\bibitem[\protect\citeauthoryear{Putter, Fiocco, and Geskus}{Putter
		et~al.}{2006}]{putter2006}
	Putter, H., M.~Fiocco, and R.~Geskus. 2006.
	\newblock Tutorial in biostatistics: Competing risks and multi-state models.
	\newblock {\em Statistics in Medicine\/}~26: 2389--2430.
	
	\bibitem[\protect\citeauthoryear{Reis and Augusto}{Reis and
		Augusto}{2015}]{Reis2015}
	Reis, P.N. and M.G. Augusto. 2015.
	\newblock What is a firm's life expectancy? empirical evidence in the context
	of Portuguese companies.
	\newblock {\em Journal of Business Valuation and Economic Loss Analysis\/}~10:
	45--75.
	
	\bibitem[\protect\citeauthoryear{Rink and Seiwert}{Rink and
		Seiwert}{2021}]{rinkseif}
	Rink, A. and I.~Seiwert. 2021.
	\newblock Aktuelle {E}ntwicklungen der {U}nternehmensdemografie.
	\newblock {\em Wirtschaft und Statistik\/}~2021/2: 41--58.
	
	\bibitem[\protect\citeauthoryear{Strohner and {Wei\ss bach}}{Strohner and
		{Wei\ss bach}}{2016}]{stroweis2015}
	Strohner, B. and R.~{Wei\ss bach}. 2016.
	\newblock Age-specific cross-sectional analysis of the fertility in
	Mecklenburg-West Pomerania with the EM algorithm (in German).
	\newblock {\em AStA Wirtschafts- und Sozialstatistisches Archiv\/}~10: 269--288.
	
	\bibitem[\protect\citeauthoryear{Toparkus and {Wei\ss bach}}{Toparkus and
		{Wei\ss bach}}{2025}]{topaweis2024}
	Toparkus, A.M. and R.~{Wei\ss bach}. 2025.
	\newblock Testing truncation dependence: The {G}umbel-{B}arnett copula.
	\newblock {\em Journal of Statistical Planning and Inference\/}~234: 106194.
	
	\bibitem[\protect\citeauthoryear{Ullah, Naimi, and {Md Yusoff}}{Ullah
		et~al.}{2016}]{salim2016}
	Ullah, M., N.~Naimi, and R.B. {Md Yusoff}. 2016.
	\newblock Are small and medium enterprises ({SME}s) in {L}ahore failing at the
	rate suggested in prior studies? an analysis of the degree of financial
	stress on small and medium enterprises and its impact on their life
	expectancy.
	\newblock {\em International Business Management\/}~10: 4258--4267.
	
	\bibitem[\protect\citeauthoryear{van~der Vaart}{van~der
		Vaart}{1998}]{vanderVaart}
	van~der Vaart, A. 1998.
	\newblock {\em Asymptotic Statistics}.
	\newblock Cambridge University Press, Cambridge.
	
	\bibitem[\protect\citeauthoryear{Vaupel}{Vaupel}{2011}]{vaupel1970}
	Vaupel, J. 2011.
	\newblock Biodemography of human aging.
	\newblock {\em Nature\/}~464: 536--542.
	
	\bibitem[\protect\citeauthoryear{{von Elm}, Altman, Egger, Pocock, Gotsche, and
		Vandenbrouck}{{von Elm} et~al.}{2007}]{strobe}
	{von Elm}, E., D.~Altman, M.~Egger, S.~Pocock, P.~Gotsche, and J.~Vandenbrouck.
	2007.
	\newblock Strengthening the reporting of observational studies in epidemiology
	(strobe) statement: guidelines for reporting observational studies.
	\newblock {\em Lancet\/}~370: 1453--1456.
	
	\bibitem[\protect\citeauthoryear{Weißbach, Dörre, Wied, Doblhammer, and
		Fink}{Weißbach et~al.}{2024}]{weissbachm2021effect}
	Weißbach, R., A.~Dörre, D.~Wied, G.~Doblhammer, and A.~Fink. 2024.
	\newblock Left-truncated health insurance claims data: Theoretical review and
	empirical application.
	\newblock {\em AStA Advances in Statistical Analysis\/}~108: 31--68.
	
	\bibitem[\protect\citeauthoryear{{Wei\ss bach} and Wied}{{Wei\ss bach} and
		Wied}{2022}]{weiswied2021}
	{Wei\ss bach}, R. and D.~Wied. 2022.
	\newblock Truncating the exponential with a uniform distribution.
	\newblock {\em Statistical Papers\/}~63: 1247--1270.
	
	\bibitem[\protect\citeauthoryear{Yadav and Lewis}{Yadav and
		Lewis}{2021}]{yadav2021}
	Yadav, K. and R.~Lewis. 2021.
	\newblock Immortal time bias in observational studies.
	\newblock {\em Journal of the American Medical Association\/}~325: 686--687.
	
\end{thebibliography}

\appendix
\renewcommand{\theequation}{\thesection.\arabic{equation}}
\numberwithin{equation}{section}

	\section{\bf $_TN$-compensator (Theorem \ref{comptn})} \label{proofcomptn}
	
	Note first that for measurable $Z$ the expectation conditional on an event is \cite[][Definition 8.9]{klenke}
	\begin{equation} \label{2punkt47}
	E_{\theta}^{\mathbb{T}}(Z) = \sum_{z} z \tilde{P}_{\theta}\{Z=z \vert \mathbb{T}\} 
	= \frac{E_{\theta}(\mathds{1}_{\mathbb{T}} Z)}{\tilde{P}_{\theta}(\mathbb{T})}, 
	\end{equation}
	where $E_{\theta}^{\mathbb{T}}(Z):= \int Z d \tilde{P}_{\theta}^{\mathbb{T}}$.
		Of course $M:=N-A$ is a martingale with respect to $P_{\theta_0}$ and filtration $\{\mathcal{G}_x: x \in \mathbb{N}_0\}$. We will infer thereof that also, for $\tilde{P}_{\theta_0}(\mathbb{T})>0$, is $\leftindex_TM= \leftindex_T N - \leftindex_T A$ a martingale with respect to $\tilde{P}_{\theta_0}^{\mathbb{T}}$ and the filtration $\{\leftindex_T{\mathcal{G}}_x: x \in \mathbb{N}_0\}$.
		
		We need to show \cite[][Section 35]{Billing} that (i) $_TM$ is adapted to $\{\leftindex_T{\mathcal{G}}_x: x \geq 0\}$, (ii) $E_{\theta}^{\mathbb{T}}(\vert \leftindex_TM(x) \vert) < \infty$, for all $x  \in \mathbb{N}_0$ and (iii) $E_{\theta}^{\mathbb{T}}(\leftindex_TM(x+v) \vert \leftindex_T{\mathcal{G}}_x)=\leftindex_TM(x)$ for all $x \in \mathbb{N}$ and $v \in \mathbb{N}$. 
		
Note for (i) first that by definition of $\leftindex_TN$ as well as of \eqref{con2d1d} it is $\leftindex_TA(x, \theta_0)= A(x, \theta_0) - A(x \wedge T, \theta_0) =\sum_{k=1}^{x} \mathds{1}_{\{k \ge T+1 \}} E_{\theta}^{\mathbb{T}}\{\Delta N(k) \vert \mathcal{G}_{k-1}\}$. 

The process $\leftindex_TN(x)$ is obviously measurable with respect to $\{\leftindex_T{\mathcal{G}}_x: x \geq 0\}$. Furthermore follows from the $\mathcal{F}_k$-measurability of $E_{\theta}^{\mathbb{T}}\{\Delta N(k) \vert \mathcal{F}_{k-1}\}$ and the subset property	
	\[
	 \mathcal{F}_k \subseteq \mathcal{G}_k \subseteq \, \leftindex_T{\mathcal{G}}_k \subseteq \, \leftindex_T{\mathcal{G}}_x
	\]
	the  $\leftindex_T{\mathcal{G}}_x$-measurability of $E_{\theta}^{\mathbb{T}}\{\Delta N(k) \vert \mathcal{F}_{k-1}\}$ for $k\leq x$. It leave to show the measurability of $\mathds{1}_{\{k\ge T+1\}}$. As $T$ is a $\mathcal{G}_x$-stopping time, follows from	
	\[
	\{k\ge T+1\} \cap \{T \leq k\}= \{T \leq k-1\}\in \mathcal{G}_k \subseteq \mathcal{G}_x \; \text{with} \; k \leq x
	\]
	the measurability of $\{k\ge T+1\}$ with respect to $\mathcal{G}_T$ due to definition \eqref{gtdef} (and especially the $\mathcal{G}_T$-measurability of $T$ \cite[][after Formula 35.20]{Billing}) and finally also the $\leftindex_T{\mathcal{G}}_x$ measurability. The process $\leftindex_TM$ is hence adapted to $\leftindex_T{\mathcal{G}}_x$. 

The property (ii) $E_{\theta}^{\mathbb{T}}\vert\leftindex_TM(x)\vert < \infty, \, \forall x \in \mathbb{N}_0$ holds because $M$ is a $(\mathcal{G}_x, \tilde{P}_{\theta}$)-martingale $E_{\theta}\vert\leftindex_TM(x)\vert =  E_{\theta}\vert M(x) - M(x \wedge T)\vert  
		\leq E_{\theta} \{\vert M(x)\vert + \vert M(x \wedge T)\vert \} 
		 =  E_{\theta}\vert M(x)\vert + \mathrm{E}\vert M(x \wedge T)\vert$ 
with all finite summands in the last expression (as $M(x \wedge T)$ is a martingale because $T$ is a stopping time by (Doob's) optional sampling theorem (see  \citet[][Theorem 2.1]{doob1953}, \citet[][Corollary to Thorem 9.3.3]{Chung} or \citet[][Theorem 35.2]{Billing}). Finally by \eqref{2punkt47} it is: 
\[
E_{\theta}^{\mathbb{T}}\vert\leftindex_TM(x)\vert = \tilde{P}_{\theta} (\mathbb{T})^{-1}  E_{\theta}\vert \mathds{1}_{\mathbb{T}}\leftindex_TM(x)\vert \le \tilde{P}_{\theta} (\mathbb{T})^{-1}  E_{\theta}\vert \leftindex_TM(x)\vert
\]
 
For (iii) by (i), the definition of the conditional expectation and because $\tilde{P}_{\theta}(X\ge T +1)>0$ it is:
\begin{eqnarray}
 & & 	\mathrm{E}_{\theta}^{\mathbb{T}}\{\leftindex_TM(x+v) \vert \leftindex_T{\mathcal{G}}_x\}= \leftindex_TM(x) 
	\Leftrightarrow  \mathrm{E}_{\theta}^{\mathbb{T}} \bigl\{\leftindex_TM(x+v)- \leftindex_TM(x)  \vert \leftindex_T{\mathcal{G}}_x \bigr\} =0 \nonumber \\
	& \Leftrightarrow & \mathrm{E}_{\theta}^{\mathbb{T}} \Bigl\{ \mathds{1}_{\{\mathbb{G}\} }\bigl[\leftindex_TM(x+v
	)- \leftindex_TM(x)\bigr]\Bigr\}=0, \, \forall \mathbb{G} \in \leftindex_T{\mathcal{G}}_x, \label{Prop40} \nonumber \\
		& \Leftrightarrow & \tilde{P}_{\theta} (\mathbb{T})^{-1} \mathrm{E}_{\theta} \Bigl\{ \mathds{1}_{\{\mathbb{G\cap \mathbb{T}}\} }\bigl[\leftindex_TM(x+v
		)- \leftindex_TM(x)\bigr]\Bigr\}=0 \nonumber \\
	& 	\Leftrightarrow & \mathrm{E}_{\theta} \Bigl\{ \mathds{1}_{\{\mathbb{G\cap \mathbb{T}}\} }\bigl[\leftindex_TM(x+v
	)- \leftindex_TM(x)\bigr]\Bigr\}=0 , \, \forall \mathbb{G} \in \leftindex_T{\mathcal{G}}_x \label{Prop41}	
\end{eqnarray}

Because $\{T \le x\} \in \leftindex_T{\mathcal{G}} \subset \leftindex_T{\mathcal{G}}_x$ and with it the complement $\{T \ge x+1\}$ those two sets are a partition of $\leftindex_T{\mathcal{G}}_x$ and we may restrict to prove \eqref{Prop41} to (a) $\mathbb{G} \subseteq \{T \leq x\}$ and (b) $\mathbb{G} \subseteq \{T \geq x+1\}$ (as any set that `overlaps' can be decomposed into two respective subsets).

\noindent {Case (a)}: We can rewrite $ \mathbb{G} \cap \mathbb{T}=(\mathbb{T}\cap \{T \leq x\})\cap(\mathbb{G} \cap \{T \leq x\})$.
		
		Now, as obviously $\mathbb{T} \in \mathcal{G}_T$, so that by definition \eqref{gtdef} it is $\mathbb{T} \cap \{T \leq x\} \in \mathcal{G}_x$. Furthermore (see \eqref{Prop41}) holds $\mathbb{G} \in  \leftindex_T{\mathcal{G}}_x$ and hence (wlog) either $\mathbb{G} \in \mathcal{G}_T$ or $\mathbb{G} \in \mathcal{G}_x$. In the first case, by definition, $\mathbb{G} \cap \{T \leq x\} \in \mathcal{G}_x$ and in the second case by the nestedness $\mathbb{G} \cap \{T \leq x\}=\mathbb{G} \in \mathcal{G}_x$. So finally $\mathbb{G}  \cap \mathbb{T} \in \mathcal{G}_x$. 
		
	Furthermore holds on the set $\{T \leq x\}$
	\begin{eqnarray*}
		\leftindex_TM(x+v)- \leftindex_TM(x) & = & M(x+v)-M(T)-M(x)+M(T) \\
		&=& M(x+v)-M(x)
	\end{eqnarray*}
so that
	\begin{equation*}
		\mathrm{E}_{\theta} \Bigl\{ \mathds{1}_{\{T \leq x\}}\bigl[\leftindex_TM(x+v)- \leftindex_TM(x)\bigr] \vert \mathcal{G}_x \Bigr\} =\mathrm{E}_{\theta} \bigl\{  M(x+v)- M(x) \vert \mathcal{G}_x \bigr\} =0.
	\end{equation*}	
by the martingale property of $M$. The definition of the conditional expectation requires
		\begin{multline*}
						\mathrm{E}_{\theta} \Bigl\{ \mathds{1}_{\{\{T \leq x\} \cap \mathbb{G} \cap \mathbb{T}\} }\bigl[\leftindex_TM(x+v)- \leftindex_TM(x)\bigr]\Bigr\}=0  \\
			\Rightarrow	 \mathrm{E}_{\theta} \Bigl\{ \mathds{1}_{\{\mathbb{G} \cap \mathbb{T}\} }\bigl[\leftindex_TM(x+v)- \leftindex_TM(x)\bigr]\Bigr\}=0 
		\end{multline*}
where for the last step remember $\mathbb{G} \subseteq \{T \leq x\}$.
		
\noindent {Case (b)}:	We start by defining set $\mathbb{K}:=\{T \leq x+v\}$. Because on event $\mathbb{G}$ holds $M(x \wedge T)=M(x)$, we may rewrite \eqref{Prop41} as follows:
				\begin{align*}
			\Leftrightarrow &\mathrm{E}_{\theta}\Bigl\{\mathds{1}_{\{\mathbb{G} \cap \mathbb{T}\}} \bigl[ \, M(x+v) - M((x+v) \wedge T) - M(x) + M(x \wedge T) \bigr]  \Bigr\} =0  \\
			\Leftrightarrow &\mathrm{E}_{\theta}\Bigl\{\mathds{1}_{\{\mathbb{G} \cap \mathbb{T}\}} \bigl[ \, M(x+v) - M((x+v) \wedge T) \bigr]  \Bigr\}=0
		\end{align*}
		which is true if additionally $\omega$ is from $\mathbb{K}^c:=\{T\ge x+v+1\}$, beause then $M((x+v)\wedge T)=M(x+v)$, and hence equality \eqref{Prop41} is fulfilled. It remains to show the equality		
		\begin{align}
			\mathrm{E}_{\theta}\Bigl\{\mathds{1}_{\{\mathbb{K} \cap \mathbb{G} \cap \mathbb{T}\}} \bigl[ \, M(x+v) - M((x+v) \wedge T) \bigr]  \Bigr\}=0.
		\end{align}		
We note first that $\mathbb{G} \cap \{T \leq x\} = \emptyset \in \mathcal{G}_x$ and hence  $\mathbb{G} \in \mathcal{G}_T$. Already used is that $\mathbb{T} \in \mathcal{G}_T$ so that, due $\mathcal{G}_T$ being a $\sigma$-field, $\mathbb{G} \cap \mathbb{T} \in \mathcal{G}_T$, i.e. by definition of $\mathcal{G}_T$ holds $\mathbb{G} \cap \mathbb{T} \cap \{T \leq x\} \in \mathcal{G}_x$. Furthermore it is $\{T \leq x\} \subseteq \{T \leq x+v\}=\mathbb{K}$ so that	
				 $\mathbb{G} \cap \mathbb{T}\cap \mathbb{K}\cap  \{T \le x\} = \mathbb{G} \cap \mathbb{T} \cap \{T \le x\} \in \mathcal{G}_x$ and hence $\mathbb{K} \cap \mathbb{G} \cap \mathbb{T} \in \mathcal{G}_T$ by the definition \eqref{gtdef}. 
		Also follows from $ \mathbb{G} \cap \mathbb{T} \in \mathcal{G}_T$ that $\mathbb{K} \cap \mathbb{G} \cap \mathbb{T} = (\mathbb{G} \cap \mathbb{T}) \cap \{ T \leq x+v\} \in \mathcal{G}_{x+v}$. In combination it is $\mathbb{K} \cap \mathbb{G} \cap \mathbb{T} \in  \mathcal{G}_{x+v} \cap \mathcal{G}_T =\mathcal{G}_{(x+v) \wedge T}$.
		With $T$ are $x+v+T$ ordered stopping times and by (the proof of) the optional sampling theorem \citep[][Theorem 35.2]{Billing} is $M((x+v) \wedge T)$ a $\mathcal{G}_{x+v}$-Martingal and it is (by definition of a conditional expectation)	
		\begin{multline*}
			\mathrm{E}_{\theta}\Bigl\{\mathds{1}_{\{\mathbb{K} \cap \mathbb{G} \cap \mathbb{T}\}} \bigl[ \, M(x+v) - M((x+v) \wedge T) \bigr]  \Bigr\} = 0 \quad \text{if the next is} \\  \mathrm{E}_{\theta}\Bigl\{ \, M(x+v) - M((x+v) \wedge T) \vert \mathcal{G}_{(x+v) \wedge T}\Bigr\} \\
			 =\mathrm{E}_{\theta}\Bigl\{ \, M(x+v)  \vert \mathcal{G}_{(x+v) \wedge T}\Bigr\} - \mathrm{E}_{\theta}\Bigl\{  M((x+v) \wedge T) \vert \mathcal{G}_{(x+v) \wedge T}\Bigr\} \\
			= M((x+v) \wedge T)- M((x+v) \wedge T)=0,
		\end{multline*}
		which proves \eqref{Prop41} now also for $\mathbb{G} \subseteq \{T  \ge x+1\}$. \qed

	\section{\bf Asymptotic variance (17)} \label{proofbedingung1}

By  \citet[][Formula 5.20]{the-statis:2002} the predictable variation process $	\langle\leftindex_TV_\cdot^c \rangle(\chi)$, i.e. the compensator of the square of $\leftindex_TV_\cdot^c (\chi):=\sum_{x=1}^{\chi}(\sqrt{n} \theta_0 (1 - \theta_0))^{-1} \Delta \leftindex_TM_\cdot^c(\chi)$ - which is a martingale itself due $\sqrt{n} \theta_0 (1 - \theta_0)$ being predictable and bounded - can be calculated as
\begin{align}
	\langle\leftindex_TV_\cdot^c \rangle(\chi)= & \sum_{x=1}^\chi \mathrm{Var}_{\theta_0}^{\mathbb{T}}\Bigl\{\Delta \bigl(\sum_{i=1}^n \leftindex_TV_i^c(x) \bigr) \vert \leftindex_T{\mathcal{F}}_{x-1}^c \Bigr\} \nonumber \\
	= & \sum_{x=1}^\chi\biggl[ \sum_{i=1}^n \mathrm{Var}_{\theta_0}^{\mathbb{T}} \big\{\Delta \leftindex_TV_i^c(x) \vert \leftindex_T{\mathcal{F}}_{x-1}^c \bigr\} \nonumber \\
	& + 2 \sum_{\substack{i_1,i_2=1\\ i_1< i_2}}^n \mathrm{Cov}_{\theta_0}^{\mathbb{T}} \bigl\{ \Delta \leftindex_TV_{i_1}^c(x),  \Delta \leftindex_TV_{i_2}^c(x) \vert \leftindex_T{\mathcal{F}}_{x-1}^c \bigr\}\biggr]. \label{varproc_base}
\end{align}
Existance requires $\langle \leftindex_TV_\cdot^c (0)\rangle=0$ and  $\langle \leftindex_TV_\cdot^c (\chi)\rangle < \infty$ \cite[][Corollary 1.4.2]{flem1991}. (Note that the condition $\leftindex_T{\mathcal{F}}_{x-1}^c$ needs a generalization from one draw to the entire sample. Details are left out here.) The first existancecondition is clear, the second follows. 

\subsection{Proof of boundedness}

Note first (without proofs here) that (a) $\mathrm{E}_{\theta_0}^{\mathbb{T}}\{\Delta \leftindex_TM_\cdot^c(x_1)  \Delta \leftindex_TM_\cdot^c(x_2) \vert \leftindex_T{\mathcal{F}}_{x_1}^c \}= \Delta \leftindex_TM_\cdot^c(x_1) \mathrm{E}_{\theta_0}^{\mathbb{T}}\{  \Delta \leftindex_TM_\cdot^c(x_2)  \vert \leftindex_T{\mathcal{F}}_{x_1}^c\}$ is zero for $x_1 < x_2$. Note secondly that, for whatever condition $A$, is it (b)  $\mathds{1}_A^2=\mathds{1}_A$. And finally note (c) that it is $ \mathrm{E}_{\theta_0}^{\mathbb{T}} \{  \Delta \leftindex_TN_{i_1}^c(x) \Delta \leftindex_TN_{i_2}^c(x) \vert  \leftindex_T{\mathcal{F}}_{x-1}^c \}=\mathrm{E}_{\theta_0}^{\mathbb{T}} \{ \Delta \leftindex_TN_{i_1}^c(x) \vert \leftindex_T{\mathcal{F}}_{x-1}^c \} \mathrm{E}_{\theta_0}^{\mathbb{T}} \{\Delta \leftindex_TN_{i_2}^c(x)  \vert  \leftindex_T{\mathcal{F}}_{x-1}^c \}$ for $i_1 \neq i_2$ due to the independence of enterprises. Consider now
\begin{equation}\label{rafael}
	\begin{split}
	\mathrm{E}_{\theta_0}^{\mathbb{T}} \{ (\leftindex_TV_\cdot^c(\chi))^2\}   =   \frac{1}{n \theta_0^2 (1- \theta_0)^2} \mathrm{E}_{\theta_0}^{\mathbb{T}}\biggl\{\Bigl(\sum_{x=1}^\chi  \Delta \leftindex_TM_\cdot^c(x) \Bigr)^2\biggr\}  \\
	 =  \frac{1}{n \theta_0^2 (1- \theta_0)^2} \biggl[ \sum_{x=1}^\chi\mathrm{E}_{\theta_0}^{\mathbb{T}}\biggl\{\bigl(  \Delta \leftindex_TM_\cdot^c(x) \bigr)^2\biggr\} \\ 
	 + 2 \sum_{\substack{x_1, x_2=1 \\
			x_1 < x_2}}^\chi \mathrm{E}_{\theta_0}^{\mathbb{T}}\biggl\{\bigl(  \Delta \leftindex_TM_\cdot^c(x_1) \bigr)\bigl( \Delta \leftindex_TM_\cdot^c(x_2) \bigr)\biggr\} \biggr] \\
	 =  \frac{1}{n \theta_0^2 (1- \theta_0)^2} \biggl[ \sum_{x=1}^\chi \mathrm{E}_{\theta_0}^{\mathbb{T}} \biggl\{ \mathrm{E}_{\theta_0}^{\mathbb{T}}\Bigl\{\bigl(  \Delta \leftindex_TM_\cdot^c(x) \bigr)^2 \vert \leftindex_T{\mathcal{F}}_{x-1}^c\Bigr\} \biggr\} 
		\end{split}
\end{equation}
In the second equality it is $\sum_{\substack{x_1, x_2=1 \\
			x_1 < x_2}}^\chi \mathrm{E}_{\theta_0}^{\mathbb{T}}\{( \Delta \leftindex_TM_\cdot^c(x_1) )( \Delta \leftindex_TM_\cdot^c(x_2) )\} \stackrel{(a)}{=} 0$ by iterated expectation. 
Again by iterated expectation it is for one summand 
\begin{equation}\label{bedingterExp}
\begin{split}
	\mathrm{E}_{\theta_0}^{\mathbb{T}}\Bigl\{\bigl(  \Delta \leftindex_TM_\cdot^c(x) \bigr)^2 \vert \leftindex_T{\mathcal{F}}_{x-1}^c\Bigr\} =  \mathrm{E}_{\theta_0}^{\mathbb{T}}\Bigl\{ \bigl( \Delta \leftindex_TN_\cdot^c(x) - \Delta \leftindex_TA_\cdot^c(x, \theta_0) \bigr) ^2 \vert \leftindex_T{\mathcal{F}}_{x-1}^c  \Bigr\}  \\
	=   \mathrm{E}_{\theta_0}^{\mathbb{T}} \bigl\{(\Delta \leftindex_TN_\cdot^c(x))^2 \vert  \leftindex_T{\mathcal{F}}_{x-1}^c\bigr\}  -  (\Delta \leftindex_TA_\cdot^c(x, \theta_0))^2, 
\end{split}	
\end{equation}
due to the $\leftindex_T{\mathcal{F}}_{x-1}^c$-predictability of $\leftindex_TA_\cdot^c(x, \theta_0)$.

Consider the first term 
\begin{multline*}
	\mathrm{E}_{\theta_0}^{\mathbb{T}} \bigl\{(\Delta \leftindex_TN_\cdot^c(x))^2 \vert  \leftindex_T{\mathcal{F}}_{x-1}^c\bigr\} 
	= \mathrm{E}_{\theta_0}^{\mathbb{T}} \Bigl\{  \sum_{i=1}^n (\Delta \leftindex_TN_i^c(x))^2 \vert  \leftindex_T{\mathcal{F}}_{x-1}^c\Bigr\} \\ +2 \mathrm{E}_{\theta_0}^{\mathbb{T}} \Bigl\{  \sum_{\substack{i_1, i_2 = \\ i_1 < i_2}}^n\bigl( \Delta \leftindex_TN_{i_1}^c(x) \Delta \leftindex_TN_{i_2}^c(x) \bigr) \vert  \leftindex_T{\mathcal{F}}_{x-1}^c\Bigr\} \\
	\stackrel{\text{by} \;(c)}{=} \mathrm{E}_{\theta_0}^{\mathbb{T}} \Bigl\{  \sum_{i=1}^n \Delta \leftindex_TN_i^c(x)^2 \vert  \leftindex_T{\mathcal{F}}_{x-1}^c\Bigr\} \quad (\text{by (b) as} \; \Delta \leftindex_TN_{i_1}^c(x) \in \{0,1\}) \\ +2 \sum_{\substack{i_1, i_2 =1 \\ i_1 < i_2}}^n \mathrm{E}_{\theta}^{\mathbb{T}} \Bigl\{ \Delta \leftindex_TN_{i_1}^c(x) \vert \leftindex_T{\mathcal{F}}_{x-1}^c\Bigr\} \mathrm{E}_{\theta_0}^{\mathbb{T}} \Bigl\{\Delta \leftindex_TN_{i_2}^c(x)  \vert  \leftindex_T{\mathcal{F}}_{x-1}^c\Bigr\} \quad (\text{by (c))}\\
	= \mathrm{E}_{\theta_0}^{\mathbb{T}} \Bigl\{  \Delta \leftindex_TN_\cdot^c(x) \vert  \leftindex_T{\mathcal{F}}_{x-1}^c\Bigr\}
	+2 \theta_0^2 \sum_{\substack{i_1, i_2 =1 \\ i_1 < i_2}}^n \leftindex_TY_{i_1}^c(x-1) \leftindex_TY_{i_2}^c(x-1)
\end{multline*}

The second term contains the number of combinations of different enterprises whose potential closure in $x$ is observable. In $x-1$ there are $\leftindex_TY_\cdot^c(x-1)$ of such enterprises so that:
\begin{multline*}
	\mathrm{E}_{\theta_0}^{\mathbb{T}} \bigl\{(\Delta \leftindex_TN_\cdot^c(x))^2 \vert  \leftindex_T{\mathcal{F}}_{x-1}^c\bigr\} 
	= \mathrm{E}_{\theta_0}^{\mathbb{T}} \Bigl\{  \Delta \leftindex_TN_\cdot^c(x) \vert  \leftindex_T{\mathcal{F}}_{x-1}^c\Bigr\}
	+2 \theta_0^2 \binom{\leftindex_TY_\cdot^c(x-1)}{2}  \\
	=\theta_0 \leftindex_TY_\cdot^c(x-1) + 2 \theta_0^2 \frac{ \leftindex_TY_\cdot^c(x-1)!}{2! (\leftindex_TY_\cdot^c(x-1) -2)!} \\
	= \theta_0 \leftindex_TY_\cdot^c(x-1) + \theta_0^2 \leftindex_TY_\cdot^c(x-1) (\leftindex_TY_\cdot^c(x-1)-1) \\
	= \theta_0 \leftindex_TY_\cdot^c(x-1) + \theta_0^2 (\leftindex_TY_\cdot^c(x-1))^2 -  \theta_0^2 \leftindex_TY_\cdot^c(x-1) \\
\end{multline*}
Insertion into in (\ref{bedingterExp}) results in:
\begin{equation*}
	\mathrm{E}_{\theta_0}^{\mathbb{T}}\Bigl\{\bigl(  \Delta \leftindex_TM_\cdot^c(x) \bigr)^2 \vert \leftindex_T{\mathcal{F}}_{x-1}^c\Bigr\} 
	=\theta_0(1-\theta_0) \leftindex_TY_i^c(x-1)
\end{equation*}
and hence \eqref{rafael} becomes:
\begin{equation*}
\frac{1}{n \theta_0 (1- \theta_0)}  \sum_{x=1}^\chi \mathrm{E}_{\theta_0}^{\mathbb{T}} \bigl\{ \leftindex_TY_\cdot^c(x-1) \bigr\} 
	 =  \frac{1}{\theta_0 (1-\theta_0)}  \sum_{x=1}^\chi  \mathrm{E}_{\theta_0}^{\mathbb{T}} \bigl\{\leftindex_TY_1^c(x-1) \bigr\}
\end{equation*}
Hence holds $\mathrm{E}_{\theta_0}^{\mathbb{T}} \{ (\leftindex_TV_\cdot^c(\chi))^2\} < \infty$ for all $\chi \in \{1, 2, \ldots\}$ because of Assumption \ref{pararaumeinfach} and hence the predictable variation process $\langle \leftindex_TV_\cdot^c \rangle$ exists.

 \subsection{Explicit representation}

For the first term in \eqref{varproc_base} it is
\begin{multline*}
	\sum_{i=1}^n \mathrm{Var}_{\theta_0}^{\mathbb{T}} \big\{\Delta \leftindex_TV_i^c(x) \vert \leftindex_T{\mathcal{F}}_{x-1}^c \bigr\} 
	= \frac{1}{n \theta_0^2 (1-\theta_0)^2} \sum_{i=1}^n \mathrm{Var}_{\theta_0}^{\mathbb{T}} \big\{\Delta \leftindex_TM_i^c(x) \vert \leftindex_T{\mathcal{F}}_{x-1}^c \bigr\} \\
	=\frac{1}{n \theta_0^2 (1-\theta_0)^2} \sum_{i=1}^n \mathrm{E}_{\theta_0}^{\mathbb{T}} \big\{ \bigl(\Delta \leftindex_TM_i^c(x) \bigr)^2 \vert \leftindex_T{\mathcal{F}}_{x-1}^c \bigr\} \\
= (n \theta_0^2 (1-\theta_0)^2)^{-1}\sum_{i=1}^n \biggl[\theta_0 \leftindex_TY_i^c(x-1) 
 - 2 \theta_0^2 \leftindex_TY_i^c(x-1) + \theta_0^2 \leftindex_TY_i^c(x-1) \biggr] \\
= \frac{1}{n\theta_0 (1-\theta_0)} \sum_{i=1}^n \leftindex_TY_i^c(x-1),
\end{multline*}
using for the third equality that $\mathrm{E}_{\theta_0}^{\mathbb{T}} \{ \Delta \leftindex_TN_i^c(x) \vert \leftindex_T{\mathcal{F}}_{x-1}^c \}=\theta_0 \leftindex_T{Y_i^c(x-1)}$ and $\leftindex_TY_i^c(x-1) \in \{0,1\}$.

For the second term of \eqref{varproc_base} holds, due to $\mathrm{E}_{\theta_0}^{\mathbb{T}}\{\Delta \leftindex_TM_i^c(x)\vert \leftindex_T{\mathcal{F}}_{x-1}^c\}=0$, that
\begin{multline*}
2 \sum_{\substack{i_1,i_2=1\\ i_1< i_2}}^n \mathrm{Cov}_{\theta_0}^{\mathbb{T}} \bigl\{ \Delta \leftindex_TV_{i_1}^c(x),  \Delta \leftindex_TV_{i_2}^c(x) \vert \leftindex_T{\mathcal{F}}_{x-1}^c \bigr\} \\
	= 2(n \theta_0^2 (1-\theta_0)^2)^{-1} \sum_{\substack{i_1,i_2=1\\ i_1< i_2}}^n \mathrm{E}_{\theta_0}^{\mathbb{T}} \bigl\{ \Delta \leftindex_TM_{i_1}^c(x) \Delta \leftindex_TM_{i_2}^c(x) \vert \leftindex_T{\mathcal{F}}_{x-1}^c \bigr\}=0,
\end{multline*} 
due to the independence in Assumption \ref{A3:Ind}.
The final argument is that by independence $\frac{1}{n} \sum_{i=1}^n \leftindex_TY_i^c(x-1) \xrightarrow{p}  \mathrm{E}_{\theta_0}\{\leftindex_TY_1^c(x-1)\}$.
\qed

\end{document}